\documentclass[preprint,12pt]{elsarticle}

\usepackage[T1]{fontenc}
\usepackage{lmodern}
\usepackage{silence}
\usepackage{microtype}
\usepackage{amsmath,amssymb,amsthm,mathtools}
\usepackage{aliascnt}
\usepackage{booktabs}
\usepackage{array}
\usepackage{tabularx}
\usepackage{enumitem}
\usepackage{graphicx}
\usepackage{subcaption}
\usepackage{float}
\usepackage{placeins}
\usepackage{listings}
\usepackage[hidelinks]{hyperref}
\usepackage[nameinlink,noabbrev]{cleveref}

\usepackage{xcolor}
\usepackage{tikz}

\newcommand{\segmentA}{\draw (0.2,2.2)--(0.8,2.2);}
\newcommand{\segmentB}{\draw (0.9,2.1)--(0.9,1.3);}
\newcommand{\segmentC}{\draw (0.9,1.1)--(0.9,0.3);}
\newcommand{\segmentD}{\draw (0.2,0.2)--(0.8,0.2);}
\newcommand{\segmentE}{\draw (0.1,0.3)--(0.1,1.1);}
\newcommand{\segmentF}{\draw (0.1,1.3)--(0.1,2.1);}
\newcommand{\segmentG}{\draw (0.2,1.2)--(0.8,1.2);}

\newcommand{\vectordigit}[1]{%
  \begin{tikzpicture}[
    x=0.55ex,
    y=0.55ex,
    baseline=-0.15ex,
    line width=0.13ex,
    line cap=round,
    draw=gray
  ]
    \ifcase#1
      \segmentA\segmentB\segmentC\segmentD\segmentE\segmentF
    \or
      \segmentB\segmentC
    \or
      \segmentA\segmentB\segmentG\segmentE\segmentD
    \or
      \segmentA\segmentB\segmentG\segmentC\segmentD
    \or
      \segmentF\segmentG\segmentB\segmentC
    \or
      \segmentA\segmentF\segmentG\segmentC\segmentD
    \or
      \segmentA\segmentF\segmentG\segmentE\segmentC\segmentD
    \or
      \segmentA\segmentB\segmentC
    \or
      \segmentA\segmentB\segmentC\segmentD\segmentE\segmentF\segmentG
    \or
      \segmentA\segmentB\segmentC\segmentD\segmentF\segmentG
    \fi
    \path[use as bounding box] (0,0)--(1,2.4);
  \end{tikzpicture}%
}

\ExplSyntaxOn
\cs_new_protected:Npn \draw_vector_line_number:n #1
  {
    \tl_map_inline:nn {#1}
      {
        \vectordigit{##1}
        \kern 0.12ex
      }
  }
\NewDocumentCommand{\vectorlinenumber}{m}
  {
    \exp_args:Ne \draw_vector_line_number:n {#1}
  }
\ExplSyntaxOff

\lstdefinestyle{pythoncopyable}{
  language=Python,
  basicstyle=\ttfamily\fontsize{8.5}{8.4}\selectfont,
  numbers=left,
  numberstyle=\vectorlinenumber,
  numbersep=10pt,
  stepnumber=1,
  frame=single,
  framesep=4pt,
  breaklines=true,
  breakatwhitespace=false,
  showstringspaces=false,
  keepspaces=true,
  columns=fullflexible,
  tabsize=4,
  upquote=true,
  keywordstyle=\color{blue!70!black},
  commentstyle=\color{green!40!black},
  stringstyle=\color{red!60!black},
  aboveskip=0.2em,
  belowskip=0.2em
}

\newif\ifarxiv
\arxivfalse

\graphicspath{{figures/}}
\journal{Theoretical Computer Science}

\newtheorem{theorem}{Theorem}
\newaliascnt{lemma}{theorem}
\newtheorem{lemma}[lemma]{Lemma}
\aliascntresetthe{lemma}
\newaliascnt{proposition}{theorem}
\newtheorem{proposition}[proposition]{Proposition}
\aliascntresetthe{proposition}
\newaliascnt{corollary}{theorem}
\newtheorem{corollary}[corollary]{Corollary}
\aliascntresetthe{corollary}
\theoremstyle{definition}
\newaliascnt{definition}{theorem}
\newtheorem{definition}[definition]{Definition}
\aliascntresetthe{definition}
\theoremstyle{remark}
\newaliascnt{remark}{theorem}

\aliascntresetthe{remark}

\crefname{lemma}{lemma}{lemmas}
\Crefname{lemma}{Lemma}{Lemmas}
\crefname{proposition}{proposition}{propositions}
\Crefname{proposition}{Proposition}{Propositions}
\crefname{corollary}{corollary}{corollaries}
\Crefname{corollary}{Corollary}{Corollaries}
\crefname{definition}{definition}{definitions}
\Crefname{definition}{Definition}{Definitions}
\crefname{appsection}{}{}
\Crefname{appsection}{}{}

\newcommand{\cF}{\mathcal{F}}
\newcommand{\cB}{\mathcal{B}}
\newcommand{\cH}{\mathcal{H}}
\newcommand{\cU}{\mathcal{U}}
\newcommand{\e}{\mathrm{e}}
\newcommand{\distlang}[2]{L_{#1,#2}}
\newcommand{\fall}[2]{(#1)_{#2}}
\newcommand{\npad}{n_{\mathrm{pad}}}

\begin{document}

\begin{frontmatter}

\title{Breaking the \texorpdfstring{$4^k$}{4\string^k} Barrier for the \texorpdfstring{$k$}{k}-Distinct Language}

\author[aff1]{Ran Ben Basat
}
\ead{r.benbasat@cs.ucl.ac.uk}
\address[aff1]{Department of Computer Science, University College London,
London, United Kingdom}

\begin{abstract}
For integers $k\le n$, let $\distlang{k}{n}$ be the set of words over $[n]$ of length at most $k$ in which no symbol is repeated.
We present a nondeterministic finite automaton (NFA) of size $3.918^k n^{O(1)}$, improving on the $4^{k+o(k)}n^{O(1)}$ construction of Ben-Basat, Gabizon, and Zehavi.

Our proof organizes several classical ingredients---product automata, hashing, and coefficient estimates---into a \emph{gadget-amplification framework}:
We take the product of many copies of a small local NFA gadget, whose language is a subset of $\distlang{r}{c}$, and hash the $k$ input symbols to copies and local colors.
The hash family guarantees that, for every repetition-free input, some hash sends at most $r$ symbols to each copy such that the resulting projection in every copy is accepted by the local gadget.
Taking the nondeterministic union of the corresponding product NFAs yields a global NFA.
Amplifying a \(200\)-state gadget for $L_{6,11}$ obtained from the small Witt design $S(4,5,11)$, this framework gives a $3.967^k n^{O(1)}$-size NFA.

We then introduce the \emph{compose-and-compress} technique, which deletes the expensive middle layers of these products and replaces paths across the deleted bands with sound one-symbol shortcut transitions. We apply it twice, once for enhancing the amplification framework and again for the local gadget, obtaining the stated result.


\end{abstract}


\end{frontmatter}

\section{Introduction}

For a finite alphabet $C$ and an integer $r\ge0$, let
\[
  \distlang{r}{C}
  \triangleq \{x_1\cdots x_j:0\le j\le r,\ x_i\in C,
      \text{ and }x_i\ne x_{i'}\text{ whenever }i\ne i'\}\ .
\]
The empty word is included.  For $C=[n]$ we write $\distlang{r}{n}$.
Ben-Basat, Gabizon, and Zehavi use the name $k$-distinct language for the sublanguage of $\distlang{k}{n}$ consisting of words of length exactly $k$~\cite{BenBasatGabizonZehavi2016}.
This convention does not affect the size bounds in this paper.
Our NFA constructions have one layer of states for each input length.
If all states are accepting, they recognize $\distlang{k}{n}$; if only the states in layer $k$ are accepting, they recognize the traditional exact-length language.
We use the prefix-closed convention because it simplifies the presentation of the techniques below.

Ben-Basat, Gabizon, and Zehavi proved a lower bound with exponential base two and an upper bound of
\[
  4^k k^{O(\log^2 k)}n^{O(1)}
  =4^{k+o(k)}n^{O(1)}\ .
\]
We improve the upper-bound base in the same NFA model. Every transition in this paper reads exactly one input symbol, and the size of an automaton is $|Q|+|\Delta|$.

\subsection{Related work}

The NFA-size question for the $k$-distinct language was posed publicly on Theoretical Computer Science Stack Exchange in 2014~\cite{BenBasat2014CST}.
Ben-Basat, Gabizon, and Zehavi subsequently introduced the parameterized-automata formulation of the $k$-distinct language~\cite{BenBasatGabizonZehavi2016}.
Besides the ordinary NFA upper and lower bounds recalled above, they constructed a bounded-ambiguity NFA for approximate counting and a nondeterministic XOR automaton, and showed how such automata combine with problem-specific automata for $k$-Path, $r$-Dimensional $k$-Matching, and Module Motif.  Molina Lovett and Shallit studied the special case $k=n$, in which the exact-length language consists of all permutations of the alphabet~\cite{MolinaLovettShallit2019}. They proved asymptotically tight bounds $4^n n^{-(\lg n)/4+\Theta(1)}$ for the minimum size of a regular expression describing this language. Their work concerns regular-expression size rather than ordinary NFA size and does not provide a sub-$4^k$ NFA construction. To our knowledge, the $O^*(4^{k+o(k)})$ NFA bound of~\cite{BenBasatGabizonZehavi2016} had not been improved before the present work.\footnote{$O^*(f(k))$ is shorthand for $f(k)\cdot n^{O(1)}$.}

The algorithmic lineage is broader.  Color-coding isolates a sought solution by a random coloring and derandomizes the choice with perfect hash families~\cite{AlonYusterZwick1995}. Divide-and-color recursively separates a solution into randomly colored parts~\cite{KneisMolleRichterRossmanith2006}, whereas Koutis introduced the algebraic multilinear-monomial method and obtained an $O^*(2^{3k/2})$ algorithm for general $k$-Path~\cite{Koutis2008}, which Williams sharpened to a randomized $O^*(2^k)$ algorithm~\cite{Williams2009}; narrow sieves extend this viewpoint to paths and packings~\cite{BjorklundHusfeldtKaskiKoivisto2017}.  These are the three techniques from which the automata constructions of Ben-Basat, Gabizon, and Zehavi are distilled.  Representative families and their associated separating collections provide another deterministic way to retain only sets that can still be extended disjointly~\cite[Section~4.2]{FominLokshtanovPanolanSaurabh2016}; Zehavi gives a systematic account of mixtures of representative sets, narrow sieves, and divide-and-color~\cite{Zehavi2015}.

The ordinary $O^*(4^{k+o(k)})$ NFA relevant here arises from divide-and-color.
Our construction follows a different route which we call the \emph{gadget-amplification framework}.
Starting from a local NFA gadget of capacity $r$ over $c$ colors, we take the product of many copies and hash every global input symbol to one copy and one local color.
For every repetition-free input of at most $k$ symbols, a covering hash family supplies a hash that assigns at most $r$ symbols to each copy and makes each local projection one accepted by that copy.
The nondeterministic union of the corresponding product NFAs is the resulting global NFA.

We note that the improved NFA construction of this paper does not translate into a new state-of-the-art algorithm for problems like $k$-Path.
Nederlof recently gave a deterministic $2^{k+O(\sqrt{k}\log^2 k)}(n+m)\log n$-time algorithm for weighted directed $k$-Path in the word-RAM model, assuming that every edge weight fits in one word~\cite{Nederlof2026}.

\subsection{Overview of our results}
We now state our main result.
\begin{theorem}[Main theorem]\label{thm:main}
For every $1\le k\le n$, an acyclic NFA recognizing $\distlang{k}{n}$ can be constructed deterministically such that
\[
  |Q|+|\Delta|
  \le 2^{1.96992k}n^{O(1)}
  <3.918^k n^{O(1)}\ .
\]
The construction time obeys the same bound.
\end{theorem}
As mentioned above, the same transition graph, with precisely the states in layer \(k\) accepting, gives the identical bound for the exact-length-$k$ convention of~\cite{BenBasatGabizonZehavi2016}.

To prove the theorem, we develop the gadget-amplification framework in three stages.
We first construct a constant-size NFA gadget from a $(c,a,b)$-separating family $\mathcal F$.
The gadget recognizes $\distlang{a+b}{c}$ using
\[
  \sum_{j=0}^{a-1}\binom cj+|\mathcal F|+\sum_{j=0}^{b-1}\binom cj
\]
states.
Although asymptotically efficient separating families are well studied~\cite{FominLokshtanovPanolanSaurabh2016}, here we need a small construction for fixed $c,a,b$.
We identify the small Witt design $S(4,5,11)$ as an $(11,3,3)$-separating family of size $66$, which yields a compact $200$-state construction for $\distlang{6}{11}$.


The framework takes the product of many copies of a local gadget and uses a hash to assign every global input symbol to one copy and one local color.
A hash succeeds on a repetition-free input when every copy receives at most the gadget capacity and the resulting local projection is accepted by that copy.
A covering hash family supplies a successful hash for every repetition-free input of at most $k$ symbols, and we take the nondeterministic union of the corresponding product NFAs.
For the complete Witt gadget, optimizing a fixed histogram of local loads without compressing the product state space already yields an NFA of size $O^*(3.967^k)$ for $\distlang{k}{n}$.


We then refine the same framework using compose-and-compress (CaC), which combines several partial gadgets into a larger one.
A partial gadget is a layered NFA that accepts only repetition-free words and comes with a downward-closed family of \emph{certified} color sets, each of size at most $r$.
A set is certified when the gadget accepts every ordering of its elements.
Every certified set $S$ also has a certified superset of each cardinality from $|S|$ through $r$.
CaC forms the product of several partial gadgets and deletes its large middle layers while preserving soundness.

\begin{table}[t]
\centering
\small
\setlength{\tabcolsep}{5pt}
\begin{tabularx}{\textwidth}{@{}>{\hsize=1.08\hsize\raggedright\arraybackslash}X>{\hsize=.92\hsize\raggedright\arraybackslash}Xl@{}}
\toprule
Stage & State accounting & Resulting size \\
\midrule
Full-product gadget amplification with a fixed finite load histogram (\S\ref{sec:fixed-profile}) & All states of the raw product & $O^*(3.967^k)$ \\
\midrule
Growing compose-and-compress amplification with the Witt automaton (\S\ref{sec:amplification}) & Witt automaton composition while compressing the product's middle layers & $O^*(3.925^k)$ \\
\midrule
Compose-and-compress of eleven Witt gadgets, followed by growing amplification (\S\ref{sec:eleven-witt}) & An improved finite partial gadget and the same two-tail~analysis & $O^*(3.9175^k)$ \\
\bottomrule
\end{tabularx}
\caption{The three proof stages. The last gives the main theorem.}
\label{tab:three-stages}
\end{table}


We use CaC at two scales within the gadget-amplification framework.
At the growing scale, we compose many copies of the automaton and delete their central product layers; this yields the intermediate-block NFA and an $O^*(3.925^k)$ bound.
At the finite scale, we compose eleven Witt gadgets and delete the middle layers of their product to obtain a better starting partial gadget.
Applying the growing-scale construction to that gadget gives the final $O^*(3.918^k)$ bound.
\Cref{tab:three-stages} summarizes the progression.



The main proof analyzes random hash functions because this gives the cleanest exponent calculation. \Cref{app:derandomization} applies known explicit splitter and restriction-family constructions to replace both random families by explicit splitter and restriction families. The resulting deterministic procedure lists all states and ordinary transitions in $2^{1.96992k}n^{O(1)}=O^*(3.918^k)$ time.

\section{Distinct-word gadgets and raw products}\label{sec:gadgets}

An NFA over an alphabet $\Sigma$ is a tuple $(Q,\Delta,q_0,F)$ with $\Delta\subseteq Q\times\Sigma\times Q$.  It is \emph{layered} if $Q$ is partitioned into $Q_0,\ldots,Q_r$, the initial state lies in $Q_0$, and every transition goes from $Q_i$ to $Q_{i+1}$.  A state in $Q_i$ has rank $i$.
Every layered NFA is acyclic.

\begin{definition}[Complete gadget]\label{def:gadget}
Let $C$ be a finite color set and let $0\le r\le |C|$.  A \emph{complete gadget of capacity $r$ over $C$} is a layered NFA with layers $Q_0,\ldots,Q_r$, a unique initial state in $Q_0$, all states accepting, and language exactly $\distlang{r}{C}$.
\end{definition}

That is, a complete gadget accepts precisely the repetition-free words of length at most $r$.
For a layered NFA $H$ with layers $Q_0,\ldots,Q_r$, we \mbox{define its \textit{state polynomial} by}
\[
  A_H(z)\triangleq\sum_{j=0}^r |Q_j|z^j\ .
\]
In particular, the number of states, denoted $s(H)$, satisfies
\[
s(H)= A_H(1)\ .
\]

We call $H$ \emph{symmetric} if $|Q_j|=|Q_{r-j}|$ for every $j$.

For a \textit{complete} gadget $G$ of capacity $r$ over a $c$-element color set, we also define its \emph{set polynomial} by
\[
  B_G(z)\triangleq\sum_{j=0}^r\binom cj z^j\ .
\]
The coefficient of $z^j$ is the number of $j$-element color sets, and every ordering of each such set is accepted by $G$. Thus, $A_G$ counts states by layer, whereas $B_G$ counts repetition-free color sets by cardinality.

For $1\le i\le m$, let $G_i$ be a complete gadget over $C_i$ with capacity $r_i$.  Assume that the palettes $C_i$ are disjoint.  Their \emph{raw product} has a state set $Q(G_1)\times\cdots\times Q(G_m)$, all accepting, with the initial state being the tuple of initial states of the automata.  Reading a color from $C_i$ moves only the $i$th coordinate.  A set $S\subseteq (\bigcup_{i=1}^m C_i)$ is \emph{product-feasible} if
\[
  |S\cap C_i|\le r_i
  \qquad(1\le i\le m)\ .
\]

\begin{lemma}[Raw product]\label{lem:raw-product}
The raw product accepts every ordering of every product-feasible set and accepts no word with a repeated color.  Its state polynomial is
\[
  \prod_{i=1}^m A_{G_i}(z)\ ,
\]
and the polynomial counting product-feasible sets by cardinality is
\[
  \prod_{i=1}^m B_{G_i}(z)\ .
\]
\end{lemma}

\begin{proof}
A product state has rank equal to the sum of its local ranks, which gives the state polynomial.  A product-feasible set is the disjoint union of one color set of size at most $r_i$ from each palette, which gives the second polynomial.
Any prescribed ordering can be read by interleaving the corresponding local paths.  If a color repeats, it repeats inside its unique local palette, and the corresponding complete gadget has no path with that repeated label.
\end{proof}

As a concrete example, let $G$ be the complete gadget of capacity two over $\{0,1\}$ shown in \Cref{fig:211-complete-gadget}.  It has four states, all accepting, and state polynomial
\[
  A_G(z)=1+2z+z^2\ .
\]

\begin{figure}[H]
\centering
\includegraphics[width=0.58\textwidth]{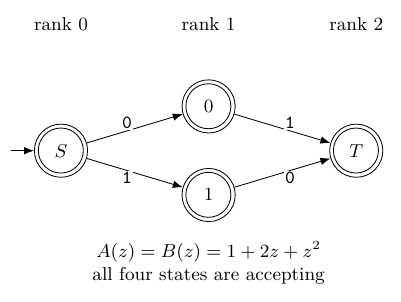}
\caption{A complete two-color gadget of capacity two. Every arrow reads one color.}
\label{fig:211-complete-gadget}
\end{figure}

Take two copies on the disjoint palettes $\{00,01\}$ and $\{10,11\}$.  A product state is an ordered pair, and each input symbol moves exactly one coordinate. Their raw product, which is complete over the four-color union, appears in \Cref{fig:211-product}.

\begin{figure}[H]
\centering
\includegraphics[width=\textwidth,trim=0 28bp 0 0,clip]{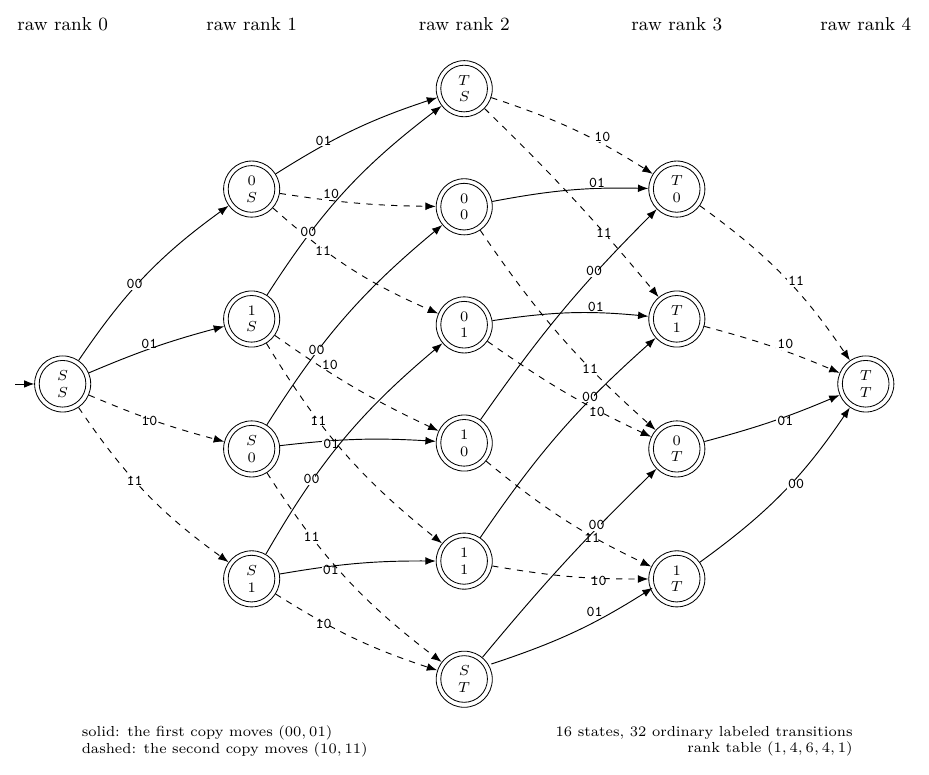}
\vspace{-0.2em}
{\small
\begin{tabular*}{0.94\textwidth}{@{\extracolsep{\fill}}ll@{}}
solid: the first copy moves $(00,01)$ &
$16$ states, $32$ transitions \\
dashed: the second copy moves $(10,11)$ &
rank table $(1,4,6,4,1)$
\end{tabular*}}
\caption{The raw product of two complete two-color gadgets.
Solid edges move the first copy and dashed edges move the second. The graph has $16$ states and $32$ \mbox{ordinary labeled transitions.}}
\label{fig:211-product}
\end{figure}

The product polynomial is
\[
  A_G(z)^2=(1+2z+z^2)^2=1+4z+6z^2+4z^3+z^4\ .
\]
Thus its coefficients $1,4,6,4,1$ are exactly the numbers of product states in ranks $0,1,2,3,4$, respectively, as visible in \Cref{fig:211-product}.

The raw product need not itself be a complete gadget over the union of the palettes: a repetition-free word can overload one palette.  The smallest symmetric example is shown in \Cref{fig:raw-product-partial}.  Each local gadget has two colors and capacity one.  The product accepts at most one color from each palette, so it accepts $00\,10$ but rejects the repetition-free word $00\,01$.  This distinction is the reason that, when we later introduce the compose-and-compress technique, we record which color sets are guaranteed to be readable.

\begin{figure}[H]
\centering
\begin{subfigure}[t]{0.29\textwidth}
  \centering
  \vspace{0pt}
  \includegraphics[width=\linewidth]{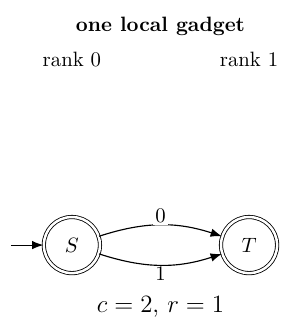}
  \caption{One complete gadget with $c=2$ and $r=1$.}
  \label{fig:raw-product-local}
\end{subfigure}
\hfill
\begin{subfigure}[t]{0.67\textwidth}
  \centering
  \vspace{0pt}
  \includegraphics[width=\linewidth]{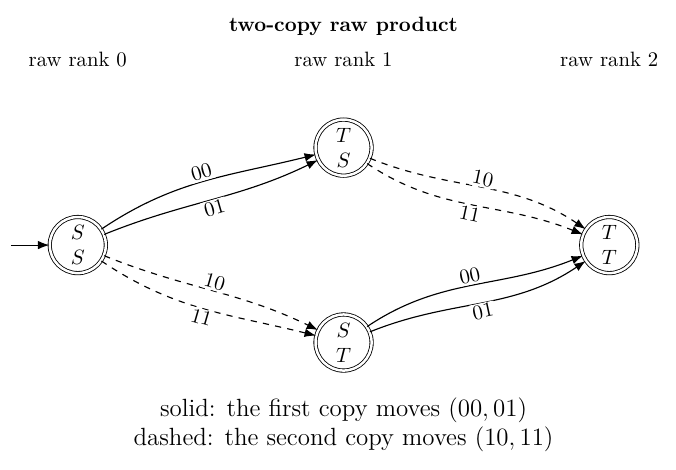}
  \caption{The two-copy raw product, viewed as a partial gadget. Solid edges move the first copy; dashed edges move the second.}
  \label{fig:raw-product-partial}
\end{subfigure}
\caption{A complete gadget and its two-copy raw product. The product has four states with rank table $(1,2,1)$. It accepts $00\,10$, but rejects $00\,01$: both colors of the latter word belong to the first palette and therefore overload its capacity-one gadget.}
\label{fig:raw-product-not-complete}

\end{figure}

\section{Separating families and the Witt gadget}\label{sec:separators}

\subsection{A general separator construction}

\begin{definition}[$(c,a,b)$-separating family]\label{def:separator}
Let $C$ be a color set of size $c$ and let $a,b\ge 1$. A family $\cF\subseteq2^C$ is a $(c,a,b)$-separating family if, for every two disjoint sets $P,S\subseteq C$ with $|P|=a$ and $|S|=b$, some $W\in\cF$ satisfies
\[
  W\cap P=\varnothing,
  \qquad
  S\subseteq W\ .
\]
Thus $W$ avoids the colors in $P$ and contains the colors in $S$.
\end{definition}

\begin{proposition}[Complete gadget from a separator]\label{prop:separator-gadget}
Suppose $c\ge a+b$ and $\cF$ is a $(c,a,b)$-separating family. Then there is a layered NFA $G$ over $C$ whose layer sizes are
\begin{equation}\label{eq:separator-profile}
  \binom c0,\binom c1,\ldots,\binom c{a-1},
  |\cF|,
  \binom c{b-1},\ldots,\binom c1,\binom c0\ .
\end{equation}
The automaton accepts exactly $\distlang{a+b}{C}$, so it is a complete gadget of capacity $a+b$. Its set polynomial is
\[
  B_G(z)\triangleq\sum_{j=0}^{a+b}\binom cj z^j\ .
\]
\end{proposition}

\begin{proof}
During the first $a-1$ positions, the state stores the set of colors already read. When the $a$th color is read, the NFA nondeterministically guesses a member $W\in\cF$ that avoids all $a$ prefix colors. From a middle-layer state $W$, upon reading $x \in W$, the NFA may move to any $(b - 1)$-set $R \subseteq W \setminus \{x\}$.
Thereafter, from a state $R$, reading $x \in R$ moves to $R \setminus \{x\}$. This gives the layer sizes in \eqref{eq:separator-profile}.

Every path is repetition-free: the prefix is stored as a set, the guessed member of $\cF$ avoids the whole prefix, and the suffix is consumed from a set contained in that member. Conversely, take a repetition-free word of length $a+b$. Apply the separating property with its prefix as $P$ and its suffix as $S$. The resulting $W$ gives a path for the word. A shorter repetition-free word can be extended to length $a+b$, so it appears as a prefix of such a path. Since every state is accepting, the automaton recognizes exactly $\distlang{a+b}{C}$.
\end{proof}



\Cref{prop:separator-gadget} gives a general gadget construction from a small $(c,a,b)$-separating family. Prior work provides asymptotically efficient families when the parameters grow~\cite{FominLokshtanovPanolanSaurabh2016}. Here we keep $c,a,b$ fixed and test several ways to reduce the family size, including SAT solving, puncturing existing families to reduce $a$ and $b$, and lifting them to increase $a$ and $b$.

Among the separating families we tried, the one in the following subsection gives the smallest NFA size bound when used in our global construction.
We note that it is possible that by designing a new small $(c,a,b)$-separating family for suitable $c,a,b$ values one can produce a smaller NFA for $\distlang{k}{n}$ using our techniques.

We release the other compact separating families found in our search, together with a standalone exact verifier, as an open-source artifact~\cite{separatingfamilies2026}, as these may be of independent interest for other applications.

The proofs in this paper do not require running the verifier.

\subsection{The small Witt design}
We use the classical small Witt design $S(4,5,11)$, a collection of five-element blocks on eleven points in which every four points lie in exactly one block; see, e.g.,~\cite[Chapter~IV]{BethJungnickelLenz1999}. For completeness, we give the explicit cyclic realization used in this paper. Let $X\triangleq\mathbb Z_{11}=\{0,1,\ldots,10\}$, and consider the six base blocks
\[
\begin{array}{lll}
B_1\triangleq\{0,1,2,3,5\},&
B_2\triangleq\{0,1,2,6,9\},&
B_3\triangleq\{0,1,2,7,8\}\ ,\\
B_4\triangleq\{0,1,3,4,7\},&
B_5\triangleq\{0,1,3,6,8\},&
B_6\triangleq\{0,1,5,7,9\}\ .
\end{array}
\]
For every $i\in\{1,\ldots,6\}$ and $a\in\mathbb Z_{11}$, define $B_i+a\triangleq\{x+a\pmod {11}:x\in B_i\}$, and let $\cB\triangleq\{B_i+a:1\le i\le6,\ a\in\mathbb Z_{11}\}$. The cited construction has the property that every four-subset of $X$ lies in exactly one block of $\cB$. It has $ |\cB|=\frac{\binom{11}{4}}{\binom54}=66 $ blocks.

\begin{lemma}[Three colors against three]\label{lem:three-three}
For every pair of disjoint triples $P,S\subseteq X$, some block $W\in\cB$ contains $S$ and avoids $P$.
\end{lemma}

\begin{proof}
For each $x\in X\setminus S$, the four-set $S\cup\{x\}$ lies in a unique block. Each block containing $S$ contains exactly two points of $X\setminus S$, so exactly $8/2=4$ blocks contain $S$. Each $p\in P$ belongs to only one of these four blocks; hence the three points of $P$ spoil at most three of them.
\end{proof}

Thus the $66$ blocks form an $(11,3,3)$ separating family. The resulting automaton has the seven layers
\[
\begin{array}{c|ccccccc}
\text{layer}&0&1&2&3&4&5&6\\
\midrule
\text{count}&1&11&55&66&55&11&1\ .
\end{array}
\]

\begin{figure}[H]
\centering
\includegraphics[width=0.98\textwidth]{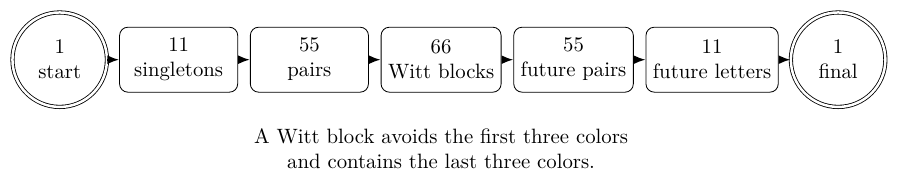}
\caption{The seven layers of the Witt gadget. A middle block avoids the first three colors and contains the final three colors.}
\label{fig:witt-layers}
\end{figure}

Applying \Cref{prop:separator-gadget} with $a=b=3$ and using \Cref{lem:three-three} yields the following.

\begin{proposition}[Witt gadget]\label{prop:witt}
The automaton obtained from the Witt design is a complete gadget of capacity six over $X$. Its polynomials are
\[
  A_W(z)=1+11z+55z^2+66z^3+55z^4+11z^5+z^6\ ,
  \]
\[
  B_W(z)=\sum_{j=0}^{6}\binom{11}{j}z^j
  =1+11z+55z^2+165z^3+330z^4+462z^5+462z^6\ .
  \]
In particular, the gadget has $A_W(1)=200$ states.
\end{proposition}

\section{The gadget-amplification framework}\label{sec:fixed-profile}

This section formally develops the gadget-amplification framework and proves the first bound below $4^k$.
Starting from a complete local gadget of capacity $r$, we take the product of many copies and use a family of hashes to assign every input symbol to one copy and one local color.
A hash succeeds on an input when each copy receives at most $r$ symbols and its resulting local projection is accepted by the gadget.
In this section, we use full-state accounting: for each hash function, we retain the first $k$ \mbox{layers of the raw product.}

The argument has six steps.
\begin{enumerate}[label=\arabic*.,leftmargin=2.4em]
  \item A hash sends each input symbol to one of $N$ copies of the gadget and to one local color.
  \item A successful hash gives each copy a repetition-free local color set whose size does not exceed the gadget capacity.
  \item We prescribe how many copies have each possible load.  This prescribed histogram describes one subset of all successful hashes.
  \item We count that subset using a multinomial coefficient and local-set counts.
  \item We use enough hashes to cover every set of at most $k$ input symbols.
  \item We multiply the number of hashes by the size of one product automaton.
\end{enumerate}

\subsection{One product automaton for one hash function}

Fix a complete gadget $G$ of capacity $r$ over a color set $C$ of size $c$.  Let $s(G)$ be its number of states.  Fix a positive integer $N$ and a function
\[
  h:[n]\longrightarrow[N]\times C\ .
\]
Write
\[
  h(x)=(b_h(x),\gamma_h(x))\ .
\]
The value $b_h(x)$ chooses one of the $N$ copies of $G$, and $\gamma_h(x)$ is the color read by that copy.  When the product automaton reads $x$, only copy $b_h(x)$ moves.  Every other copy remains in its current state.

\begin{lemma}[One fixed map]\label{lem:fixed-map}
Let $d_G$ be the maximum, over a local state and a local color, of the number of possible successor states.  For every $h:[n]\to[N]\times C$ and every cutoff $k$, there is a layered NFA $A_h$ over $[n]$ with at most $s(G)^N$ states and at most
\[
  d_G n s(G)^N
\]
ordinary transitions.  For every word of length at most $k$, $A_h$ accepts exactly when the symbol sequence received by every local copy corresponds to a path in that copy of $G$.
\end{lemma}

\begin{proof}
A product state is an $N$-tuple of local states. Its rank is the sum of the local ranks, which tracks the number of symbols processed globally.  We retain only tuples of rank at most $k$, and all states are accepting.

Suppose the current tuple is $(q_1,\ldots,q_N)$ and the next input symbol is $x$. Let $h(x)=(i,a)$. The product may replace $q_i$ by any local successor of $q_i$ under color $a$; all other coordinates are unchanged. This is exactly the raw product described above.

There are at most $s(G)^N$ tuples.  For a fixed tuple and a fixed input symbol $x$, only one coordinate moves, and it has at most $d_G$ possible successors.
Thus there are at most $d_Gns(G)^N$ transitions.  Every transition increases the total rank by one, so the automaton is layered.  The acceptance statement is immediate from the coordinatewise construction.
\end{proof}

Two simple facts will be used repeatedly.  First, if an original symbol $x$ occurs twice, both occurrences are sent to the same pair $(b_h(x),\gamma_h(x))$.  The same local copy therefore sees the same color twice.  Since every path in a complete gadget is repetition-free, the product rejects.
Second, for a set $S$ of distinct original symbols, we say that $h$ \emph{succeeds on $S$} if every copy receives a repetition-free local color set of size at most $r$; in that case the product accepts every ordering of $S$.

Write
\[
  B(z)\triangleq B_G(z)=\sum_{j=0}^{r}B_jz^j\ ,
\]
so $B_j\triangleq [z^j]B(z)$ is the coefficient of $z^j$, namely the number of local color sets of size $j$.  Fix a $k$-element set $S\subseteq[n]$ and choose $h$ uniformly from all functions $[n]\to[N]\times C$.

\begin{lemma}[Exact success probability]\label{lem:success}
The probability \mbox{that $h$ succeeds on $S$ is}
\begin{equation}\label{eq:success}
  p_{k,N}
  =\frac{k!}{(cN)^k}[z^k]B(z)^N\ .
\end{equation}
\end{lemma}

\begin{proof}
We count successful restrictions of $h$ to the labeled elements of $S$.

For each copy, choose one local color set.  If copy $i$ receives a set of size $j$, that choice contributes $B_jz^j$.  Therefore $[z^k]B(z)^N$ counts the choices of one set in each copy whose total size is exactly $k$.

Such a choice specifies exactly $k$ distinct copy-color cells.  The $k$ labeled elements of $S$ can be assigned bijectively to those cells in $k!$ ways.  Hence the number of successful hashes is $k![z^k]B(z)^N$.  The total number of hashes restricted to $S$ is $(cN)^k$, which gives \eqref{eq:success}.
\end{proof}

Next, we analyze the required number of hashes.

\begin{lemma}[A family covering all sets up to size $k$]\label{lem:cover}
For $p_{k,N}>0$, there is a family $\cH$ of
\[
  O\!\left(p_{k,N}^{-1}k\log(2n)\right)
\]
functions from $[n]$ to $[N]\times C$ such that every subset of $[n]$ of size at most $k$ succeeds for at least one function in $\cH$.
\end{lemma}

\begin{proof}
It suffices to cover every $k$-set: every smaller set can be extended to a $k$-set, and success is preserved when elements are deleted.
Choose $T$ independent random functions.  A fixed $k$-set is missed by all of them with probability at most
\[
  (1-p_{k,N})^T\le \exp(-p_{k,N}T)\ .
\]
There are $\binom nk\le n^k$ sets of size $k$.  Therefore the expected number of uncovered $k$-sets is at most
\[
  n^k\exp(-p_{k,N}T)\ .
\]
This quantity is smaller than one once
\[
  T>p_{k,N}^{-1}(k\log n+1)\ .
\]
Thus a family covering every $k$-set exists.
\end{proof}

Taking the nondeterministic union of the automata $A_h$, one for each $h\in\cH$, recognizes $\distlang{k}{n}$.  Concretely, let $q_h$ be the initial state of $A_h$.  Take disjoint copies of the $A_h$, delete each $q_h$, and introduce a new initial state $q_0$.  For every transition $q_h \xrightarrow{a} v$ in $A_h$, add $q_0 \xrightarrow{a} v$ with the same label $a$.  Retain every transition whose endpoints were not deleted.  All states are accepting, including $q_0$.  Thus every transition still reads exactly one input symbol, and no $\varepsilon$-transitions are introduced.
It remains to lower-bound the success probability using a chosen load histogram and combine the resulting hash-family bound with \Cref{lem:fixed-map}.

\subsection{Load histograms and finite templates}

A \emph{load histogram} records how many copies receive each possible number of colors. We use $K$ for the auxiliary length used in the analysis; later we choose $K=k+O(1)$.
If $N$ copies jointly receive $K$ colors, let $\eta_j$ denote the number of copies that receive exactly $j$ colors.  Necessarily
\begin{equation}\label{eq:histogram-constraints}
  \sum_{j=0}^{r}\eta_j=N,
  \qquad
  \sum_{j=0}^{r}j\eta_j=K\ .
\end{equation}
Every nonnegative integer vector satisfying \eqref{eq:histogram-constraints} is a possible load histogram.

Expanding $B(z)^N$ by histograms gives the exact identity
\begin{equation}\label{eq:all-histograms}
  [z^K]B(z)^N
  =
  \sum_{\substack{\eta_0,\ldots,\eta_r\ge0\\
                  \sum_j\eta_j=N\\
                  \sum_jj\eta_j=K}}
  \binom{N}{\eta_0,\ldots,\eta_r}
  \prod_{j=0}^{r}B_j^{\eta_j}\ .
\end{equation}
The multinomial coefficient chooses which copies have each load.  After those loads are fixed, the factor $B_j^{\eta_j}$ chooses one $j$-element local color set for every copy of load $j$.

Every summand in \eqref{eq:all-histograms} is nonnegative.  We may therefore obtain a lower bound by choosing any one feasible histogram and discarding all other successful assignments.

To specify such a histogram as the input length grows, we choose nonnegative constant integers
\[
  m_0,m_1,\ldots,m_r\ ,
\]
with $m_j>0$ only when $B_j>0$.

Think of this vector as one finite template.
Define
\[
  D\triangleq\sum_{j=0}^{r}m_j,
  \qquad
  M\triangleq\sum_{j=0}^{r}jm_j>0\ .
\]
Thus one copy of the template uses $D$ gadget copies and carries $M$ input symbols. For a positive integer $t$, repeat the template $t$ times. Then
\begin{equation}\label{eq:repeated-template}
  N=tD,
  \qquad
  K=tM,
  \qquad
  \eta_j=tm_j\ .
\end{equation}
This is why we distinguish $K$ from $k$: repeating the template yields the coefficient lower bound only for lengths $K$ divisible by $M$.  For an arbitrary target length $k$, in \Cref{subsec:coefficient-to-nfa} we let $K$ be the least multiple of $M$ with $K\ge k$, temporarily add $K-k$ fresh symbols, construct the automaton for length $K$, then delete every transition labeled by a fresh symbol, and retain only layers $0,\ldots,k$.  Since $M$ is fixed, $0\le K-k<M$, and hence $K=k+O(1)$.

It is useful to normalize the template by setting
\[
  p_j\triangleq\frac{m_j}{D},
  \qquad
  \lambda\triangleq\frac{D}{M}\ .
  \]
Write $\mathbf p\triangleq(p_0,\ldots,p_r)$.

Here $p_j$ is the fraction of copies with load $j$, while $\lambda$ is the number of gadget copies per input symbol.  The relations in \eqref{eq:repeated-template} imply
\begin{equation}\label{eq:profile-identities}
  \sum_jp_j=1,
  \qquad
  \sum_jjp_j=\frac{M}{D}=\frac1\lambda,
  \qquad
  tD=\lambda K\ .
\end{equation}
The last identity converts quantities measured per gadget copy into quantities measured per input symbol.

\begin{center}
\small
\begin{tabularx}{0.94\textwidth}{@{}lX@{}}
\toprule
symbol & meaning \\
\midrule
$m_j$ & copies of load $j$ in one finite template \\
$D=\sum_jm_j$ & total gadget copies in one template \\
$M=\sum_jjm_j$ & total input symbols carried by one template \\
$t$ & number of repetitions of the template \\
$N=tD$ & total number of gadget copies in the product \\
$K=tM$ & input length analyzed by the repeated template \\
$p_j=m_j/D$ & fraction of copies whose load is $j$ \\
$\lambda=D/M=N/K$ & number of gadget copies per input symbol \\
\bottomrule
\end{tabularx}
\end{center}

\subsection{Counting the chosen histogram, one factor at a time}\label{sec:single-hist-lb}

Substitute the histogram $\eta_j=tm_j$ into one summand of \eqref{eq:all-histograms}.  We obtain
\begin{equation}\label{eq:type-lower}
  [z^K]B(z)^N
  \ge
  \binom{tD}{tm_0,\ldots,tm_r}
  \prod_{j=0}^{r}B_j^{tm_j}\ .
\end{equation}
The multinomial coefficient assigns the loads: for every $j$, it chooses which $tm_j$ copies have load $j$.  Once the loads are fixed, each of those copies has $B_j$ choices for its local color set, giving the second factor $\prod_{j=0}^r B_j^{tm_j}$.  We now lower-bound the two factors on the right-hand side of \eqref{eq:type-lower} separately.

\paragraph{The multinomial factor}
Let $J\triangleq\{j:m_j>0\}$. For a positive integer $u$, Stirling's formula in base-two logarithmic form is
\begin{align*}
    \log_2u! &\ge u\log_2u-(\log_2\e)u\ .\\
    \log_2u! &\le u\log_2u-(\log_2\e)u+O(\log u)\ .
\end{align*}
Apply the lower bound to $(tD)!$ and the upper bound to every $(tm_j)!$ with $j\in J$. Then
\begin{equation}\label{eq:multinomial-after-stirling}
\begin{aligned}
\log_2\binom{tD}{tm_0,\ldots,tm_r}
&=\log_2(tD)!-\sum_{j\in J}\log_2(tm_j)!\\
&\ge tD\log_2(tD)-(\log_2\e)tD\\
&\quad-\sum_{j\in J}\left(
tm_j\log_2(tm_j)-(\log_2\e)tm_j+O(\log(tm_j))
\right)\\
&\ge tD\log_2(tD)-\sum_{j\in J}tm_j\log_2(tm_j)
-O(\log t)\ .
\end{aligned}
\end{equation}
The last inequality uses $\sum_{j\in J}m_j=D$ to cancel the linear terms.
Since $J$ and all positive $m_j$ are fixed, the sum of the denominator error terms is $O(\log t)$.
Now use $m_j=Dp_j$:
\begin{align*}
\sum_{j\in J}tm_j\log_2(tm_j)
&=tD\sum_{j\in J}p_j\log_2(tDp_j)\\
&=tD\log_2(tD)\sum_{j\in J}p_j
+tD\sum_{j\in J}p_j\log_2p_j\\
&=tD\log_2(tD)
+tD\sum_{j\in J}p_j\log_2p_j\ .
\end{align*}
Substituting the preceding identity into \eqref{eq:multinomial-after-stirling} gives
\begin{equation}\label{eq:multinomial-entropy}
\log_2\binom{tD}{tm_0,\ldots,tm_r}
\ge tD\left(-\sum_{j\in J}p_j\log_2p_j\right)
-O(\log t)\ .
\end{equation}

\paragraph{The local-set factor}
Taking logarithms of the second factor in \eqref{eq:type-lower},
\begin{equation}\label{eq:local-set-contribution}
\begin{aligned}
  \log_2\left(\prod_{j=0}^{r}B_j^{tm_j}\right)
  &=\sum_{j\in J}tm_j\log_2B_j\\
  &=tD\sum_{j\in J}p_j\log_2B_j\ .
\end{aligned}
\end{equation}

\paragraph{Combining the two factors}
Define
\begin{equation}\label{eq:Psi}
  \Psi_B(\mathbf p)
  \triangleq -\sum_{j:p_j>0}p_j\log_2p_j
   +\sum_{j:p_j>0}p_j\log_2B_j\ .
\end{equation}

The first sum is the per-copy logarithmic contribution from placing the loads among the copies. The second is the per-copy logarithmic contribution from choosing the local color sets after the loads are fixed.

Add \eqref{eq:multinomial-entropy} and \eqref{eq:local-set-contribution}. By the definition of $\Psi_B$,
\begin{equation}\label{eq:coefficient-before-rescale}
\log_2[z^K]B(z)^N
\ge tD\Psi_B(\mathbf p)-O(\log t)\ .
\end{equation}
Since $M$ is fixed and $K=tM$, we have $t=\Theta(K)$, and hence $O(\log t)=o(K)$. Also, \eqref{eq:profile-identities} gives $tD=\lambda K$. Therefore
\begin{equation}\label{eq:coefficient-lower-fixed}
\log_2[z^K]B(z)^N
\ge \lambda K\Psi_B(\mathbf p)-o(K)\ .
\end{equation}

\subsection{From the coefficient to the final NFA size}\label{subsec:coefficient-to-nfa}

We now translate the coefficient lower bound into the number of hashes and then combine it with the product-state cost.
Recall that
\[
  B(z)=\sum_{j=0}^{r}B_jz^j\qquad ,
  \qquad B_j=[z^j]B(z)\ .
\]
Thus $B$ is the generating polynomial for local color sets, while $B_j$ is the number of sets of the particular size $j$.

\begin{theorem}[Amplification from a fixed load histogram]\label{thm:profile}
For the choices above, there is an acyclic NFA recognizing $\distlang{k}{n}$ of size
\[
  2^{(E_{\mathrm{fix}}+o(1))k}n^{O(1)}\ ,
\]
where
\begin{equation}\label{eq:profile-exponent}
  E_{\mathrm{fix}}
  \triangleq\lambda\log_2s(G)
   +\max\left\{0,
      \log_2(c\lambda\e)-\lambda\Psi_B(\mathbf p)
    \right\}\ .
\end{equation}
\end{theorem}

\begin{proof}
We first prove the claim for lengths $K=tM$, using $N=tD=\lambda K$ copies.

\paragraph{Step 1: the cost of one product automaton}
By \Cref{lem:fixed-map}, one hash uses at most $s(G)^N$ states, and its transition count is at most $d_Gn$ times larger. Therefore the total size of one product automaton, counting both states and transitions, is at most
\begin{equation}\label{eq:one-branch-size}
  (1+d_Gn)s(G)^N
  =(1+d_Gn)2^{\lambda K\log_2s(G)}
  \le 2^{\lambda K\log_2s(G)}n^{O(1)}\ .
\end{equation}

\paragraph{Step 2: the number of hashes}
By \Cref{lem:success},
\begin{equation}\label{eq:inverse-success}
  p_{K,N}^{-1}
  =\frac{(cN)^K}{K!\,[z^K]B(z)^N}\ .
\end{equation}
We first estimate the factorial ratio $(cN)^K/K!$.  Since $N=\lambda K$, Stirling's formula gives
\begin{equation}\label{eq:factorial-ratio}
\begin{aligned}
  \log_2\frac{(cN)^K}{K!}
  &=K\log_2(c\lambda K)-\log_2K!\\
  &=K\log_2(c\lambda K)
    -K\log_2K+(\log_2\e)K+O(\log K)\\
  &=K\log_2(c\lambda\e)+O(\log K)\ .
\end{aligned}
\end{equation}
Subtracting the coefficient lower bound \eqref{eq:coefficient-lower-fixed} from \eqref{eq:factorial-ratio} yields
\[
  \log_2p_{K,N}^{-1}
  \le
  K\left(\log_2(c\lambda\e)
    -\lambda\Psi_B(\mathbf p)\right)+o(K)\ .
  \]

By \Cref{lem:cover} and the preceding inequality, there is a covering family $\cH$ with
\begin{equation}\label{eq:hash-rate}
\begin{aligned}
  |\cH|
  &=O\!\left(p_{K,N}^{-1}K\log(2n)\right)\\
  &\le 2^{\left(\max\left\{0,\log_2(c\lambda\e)-\lambda\Psi_B(\mathbf p)\right\}+o(1)\right)K}n^{O(1)}\ .
\end{aligned}
\end{equation}
The multiplicative factor $K\log(2n)$ is folded into the $n^{O(1)}$ term.

\paragraph{Step 3: combine the branches}
As explained, the final NFA is the nondeterministic union of one product automaton per hash.
Its size is therefore bounded by
\[
  (\text{number of hashes})
  \times
  (\text{size of one product automaton})\ ,
\]
up to polynomial factors.  Multiplication adds the two exponents in \eqref{eq:one-branch-size} and \eqref{eq:hash-rate}.  Their sum is exactly $E_{\mathrm{fix}}$ in \eqref{eq:profile-exponent}.  This proves the theorem for $K=tM$.

\paragraph{Step 4: remove the divisibility assumption}
For an arbitrary target length $k$, let $K$ be the least multiple of $M$ with $K\ge k$.  Add $K-k$ fresh symbols temporarily and write
\[
  \npad\triangleq n+K-k
\]
for the padded alphabet size.  Construct the automaton for $\distlang{K}{\npad}$.  Then delete every transition labeled by a fresh symbol and retain only layers $0,\ldots,k$.  The resulting automaton recognizes $\distlang{k}{n}$.  Since $M$ is fixed, $K=k+O(1)$, and since $n\ge k$, we also have $\npad=O(n)$.  The asymptotic bound is unchanged.
\end{proof}

\subsection{The Witt gadget under the fixed-histogram analysis}
To obtain a concrete NFA size bound, we apply \Cref{thm:profile} to the Witt gadget using the load template in \Cref{tab:witt-loads}.
\Cref{app:witt-histogram-selection} explains how this template was selected; that optimization is not needed to verify the bound.

\begin{table}[H]
\centering
\caption{The repeated load template for the Witt gadget.  One template uses $D=1000$ copies.  Repeating it $t$ times gives $tm_j$ copies of each displayed load.}
\label{tab:witt-loads}
\begin{tabular}{@{}rrrrrrrr@{}}
\toprule
load $j$ &0&1&2&3&4&5&6\\
\midrule
local sets $B_j=\binom{11}{j}$ &1&11&55&165&330&462&462\\
copies in one template $m_j$ &0&0&5&28&107&292&568\\
\bottomrule
\end{tabular}
\end{table}

The proof above uses the single-histogram bound from~\Cref{sec:single-hist-lb}. Because the Witt gadget has only seven possible loads, one can also evaluate the full sum in~\eqref{eq:all-histograms}; this gives the slightly better bound $O^*(3.9661^k)$. We omit the calculation because the later sections give stronger bounds.

\begin{corollary}[Fixed-histogram Witt bound]\label{cor:witt-fixed}
There is an acyclic NFA recognizing $\distlang{k}{n}$ of size
\[
  3.967^k n^{O(1)}\ .
\]
\end{corollary}

\begin{proof}
The total number of copies in one template is
\[
  D=5+28+107+292+568=1000\ .
\]
The number of input symbols carried by one template is
\begin{align*}
  M
  &=2\cdot5+3\cdot28+4\cdot107+5\cdot292+6\cdot568\\
  &=5390\ .
\end{align*}
Therefore
\[
  \lambda=\frac{D}{M}=\frac{100}{539}\ .
\]

Starting from \eqref{eq:Psi} and substituting $p_j=m_j/D$, we obtain
\begin{align*}
  \Psi_{B_W}(\mathbf p)
  &=-\sum_{j:m_j>0}\frac{m_j}{D}\log_2\frac{m_j}{D}
    +\sum_{j:m_j>0}\frac{m_j}{D}\log_2B_j\\
  &=\frac1D\sum_{j:m_j>0}m_j
    \left(\log_2D+\log_2\frac{B_j}{m_j}\right)\\
  &=\log_2D+\frac1D\sum_{j:m_j>0}m_j\log_2\frac{B_j}{m_j}\ ,
\end{align*}
where the last equality uses $\sum_jm_j=D$.
Substituting the five columns of \Cref{tab:witt-loads} gives
\[
  \Psi_{B_W}(\mathbf p)>10.25261\ .
\]
Define the exponent contributed by one product automaton by
\[
  E_{\mathrm{branch}}\triangleq\lambda\log_2 200\ .
\]
By \eqref{eq:one-branch-size}, one product automaton has size at most $2^{E_{\mathrm{branch}}K}n^{O(1)}$.  Using $\lambda=100/539$ gives
\[
  E_{\mathrm{branch}}<1.41816\ .
\]

Define the exponent contributed by the covering family by
\[
  E_{\mathrm{hash}}
  \triangleq
  \max\left\{0,
    \log_2(11\lambda\e)-\lambda\Psi_{B_W}(\mathbf p)
  \right\}\ .
\]
By \eqref{eq:hash-rate}, the number of hash branches is at most $2^{(E_{\mathrm{hash}}+o(1))K}n^{O(1)}$.  The numerical bounds
\[
  \log_2(11\lambda\e)<2.47185,
  \qquad
  \lambda\Psi_{B_W}(\mathbf p)>1.90215
\]
therefore imply
\[
  E_{\mathrm{hash}}
  <\max\{0,2.47185-1.90215\}
  =0.56970\ .
\]

Finally, \eqref{eq:profile-exponent} states that $E_{\mathrm{fix}}=E_{\mathrm{branch}}+E_{\mathrm{hash}}$.  Hence
\begin{equation}\label{eq:witt-fixed}
  E_{\mathrm{fix}}
  <1.41816+0.56970
  =1.98786,
  \qquad
  2^{E_{\mathrm{fix}}}<3.967\ .
\end{equation}
\end{proof}

This completes the fixed-histogram analysis.  It gives an $O^*(3.967^k)$-size NFA, strictly improving on the previous $O^*(4^{k+o(k)})$ bound.  Its weakness is equally clear: \eqref{eq:one-branch-size} pays for all $200^N$ product states, including the very large middle layers.  The next section introduces compose-and-compress, whose purpose is to remove those layers while preserving correctness.

\section{Compose-and-compress}\label{sec:compose}

The fixed-histogram proof counts every state of a raw product.  Its middle layers are the expensive ones.  Compose-and-compress (CaC) removes a band of product layers centered at the middle and replaces each path across that band by one ordinary one-symbol transition.  The colors read on the omitted part of the path are hidden witnesses: they justify the shortcut but are not read by the compressed automaton.

We use the operation twice.  First, the number of gadget copies grows with an intermediate block length; this creates the intermediate-block NFA used in the global construction.  Second, exactly eleven Witt complete gadgets are compressed to make a better finite partial gadget.  The growing construction is then applied again.

\subsection{The complete two-color example}\label{sec:211}
Before formalizing the procedure, let us work with an illustrative example.
The complete gadget in \Cref{fig:211-complete-gadget} is the $(2,1,1)$ separator gadget.  We now apply compose-and-compress to its two-copy raw product from \Cref{fig:211-product}.

Delete raw rank two and shift raw ranks three and four down by one.  The six middle states disappear, so the compressed state polynomial is
\[
  1+4z+4z^2+z^3\ .
\]

A new transition crosses the gap whenever the raw product has a two-step path across the deleted rank.  Either raw label may remain visible; the other is the hidden witness.  The complete ten-state compressed NFA appears in \Cref{fig:211-compressed}.  A braced label denotes two ordinary transitions, one for each displayed symbol.

\begin{figure}[H]
\centering
\includegraphics[width=\textwidth]{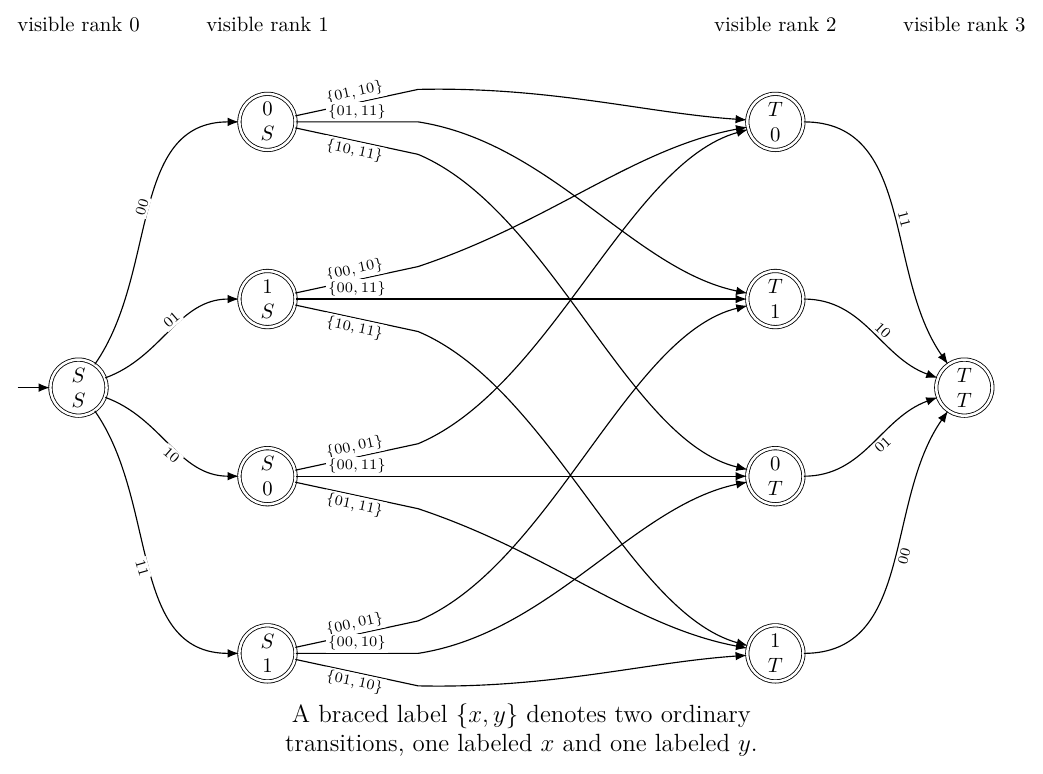}
\caption{The complete compressed NFA after composing two $(2,1,1)$ gadgets.
Every state appears once and every transition label is printed on its arrow.}
\label{fig:211-compressed}
\end{figure}

For example, the two transitions from $(0,S)$ to $(T,0)$ are justified by the raw paths
\[
 (0,S)\xrightarrow{01}(T,S)\xrightarrow{10}(T,0)
 \qquad\text{and}\qquad
 (0,S)\xrightarrow{10}(0,0)\xrightarrow{01}(T,0)\ .
\]
Thus one compressed transition is labeled $01$ and the other is labeled $10$.
There are twelve braced central arcs, hence $24$ central transitions, represented by twelve braced arcs, and eight inherited (raw) transitions.

The compressed NFA recognizes $\distlang{3}{\{00,01,10,11\}}$.  Soundness follows by expanding every shortcut to a raw path.  A repetition-free word of length two or three can be extended by one unused color; place that color in the hidden position of a raw path and compress the resulting two-step segment.  Words of length zero or one stay in the retained lower ranks.

\FloatBarrier

\subsection{Partial gadgets that can be composed again}
As \Cref{fig:raw-product-partial} illustrates, a raw product of complete gadgets need not be complete over the union of their palettes.
The preceding example remains a complete gadget because each two-color factor has capacity two: its capacity equals the size of its palette, so no distinct set can overload a factor.  If a factor has capacity $r<|C|$ and the composed gadget has target capacity $q>r$, then a set of $r+1$ distinct colors may lie entirely in one factor's palette, and that factor cannot read it.

To allow the operation to be iterated, we therefore record only the sets that are guaranteed to work.

\begin{definition}[Partial gadget]\label{def:partial-gadget}
A \emph{partial gadget of capacity $r$ over $C$} is a layered NFA $H$ with layers $Q_0,\ldots,Q_r$, a unique initial state, all states accepting, and a specified family $\cU_H\subseteq2^C$ of \emph{certified} sets. The NFA $H$ and the family $\cU_H$ must satisfy:
\begin{enumerate}[label=(P\arabic*),leftmargin=3em]
  \item $\cU_H$ is nonempty and downward closed, and every member has size at most $r$;
  \item every word accepted by $H$ is repetition-free;
  \item for every $S\in\cU_H$, every ordering of $S$ is accepted by $H$;
  \item if $S\in\cU_H$ and $0\le h\le r-|S|$, some $T\in\cU_H$ satisfies $S\subseteq T$ and $|T|=|S|+h$.
\end{enumerate}
Its certified-set polynomial is
\[
  B_H(z)\triangleq\sum_{j=0}^r
  |\{S\in\cU_H:|S|=j\}|z^j\ .
\]
The state polynomial and symmetry are exactly as for complete gadgets.
\end{definition}

Every complete gadget is a partial gadget: take all subsets of $C$ of size at most $r$ as the certified sets. The extra family is needed only after composition; it records what the next use of the construction may safely count. For example, in the partial gadget of \Cref{fig:raw-product-partial}, let $C_1=\{00,01\}$ and $C_2=\{10,11\}$ be the two palettes. Its family of certified sets is $\cU_H\triangleq\{S\subseteq C_1\cup C_2:|S\cap C_1|\le1\text{ and }|S\cap C_2|\le1\}$. Thus $\{00,10\}$ is certified, whereas $\{00,01\}$ is not: both colors in the latter set would have to be read by the capacity-one first factor.

\begin{lemma}[Product of partial gadgets]\label{lem:partial-product}
Take $m$ copies of a partial gadget $H$ on pairwise disjoint color sets.  Their raw product is a partial gadget of capacity $R=mr$.  Writing $H_i$ for the $i$th copy, its certified-set family is $\cU_{H^{\times m}}\triangleq\left\{\bigcup_{i=1}^m S_i:S_i\in\cU_{H_i}\text{ for every }i\in[m]\right\}$, and the two polynomials are
\[
  A_H(z)^m
  \qquad\text{and}\qquad
  B_H(z)^m\ .
\]
\end{lemma}

\begin{proof}
Let $C_i$ be the palette of $H_i$. Because the palettes are disjoint, every $S\subseteq\bigcup_iC_i$ decomposes uniquely as $S=\bigcup_iS_i$, where $S_i\triangleq S\cap C_i$. Because each $\cU_{H_i}$ is nonempty, choosing one $S_i\in\cU_{H_i}$ for every $i$ shows that $\cU_{H^{\times m}}$ is nonempty. If $S$ is certified and $T\subseteq S$, then $T\cap C_i\in\cU_{H_i}$ by downward closure, so $T$ is certified; also $|S|=\sum_i|S_i|\le mr$. Thus (P1) holds. Any path in the product projects to a path in each $H_i$, so a repeated color would contradict (P2) in the corresponding copy. Conversely, for a certified $S$ and any ordering $w$ of $S$, each projection $w|_{C_i}$ labels a path in $H_i$ by (P3); interleaving these paths according to $w$ gives a product path labeled $w$. For (P4), given $0\le h\le mr-|S|$, choose integers $0\le h_i\le r-|S_i|$ with $\sum_i h_i=h$. Extend each $S_i$ by $h_i$ colors using (P4) in $H_i$; the union is a certified extension of $S$ of size $|S|+h$. 

A product state is an $m$-tuple of local states. If their ranks are $j_1,\ldots,j_m$, then the product state has rank $j_1+\cdots+j_m$ and contributes $z^{j_1+\cdots+j_m}$ to the state polynomial. Multiplying the $m$ copies of $A_H(z)$ counts exactly these tuples, so the state polynomial is $A_H(z)^m$. Likewise, multiplying $B_H(z)$ chooses one certified set $S_i$ in each copy and records their total size in the exponent. Because the palettes are disjoint, $(S_1,\ldots,S_m)\mapsto\bigcup_iS_i$ is a bijection, so the certified-set polynomial is $B_H(z)^m$.
\end{proof}

\subsection{The compression operation}

Let $G$ be a symmetric partial gadget of capacity $r$ with polynomials $A(z),B(z)$.  Take $m$ copies on disjoint color sets and write
\[
  A(z)^m=\sum_{j=0}^{R}a_jz^j,
  \qquad
  B(z)^m=\sum_{j=0}^{R}b_jz^j,
  \qquad
  R\triangleq mr\ .
\]
The raw product has capacity $R$. We construct from it a gadget of a smaller target capacity $q$, where $1\le q<R$. To reduce the capacity by $R-q$, we delete that many consecutive raw layers and replace every path crossing the resulting gap by a one-symbol shortcut. Define
\[
  s\triangleq R-q,
  \qquad
  \ell\triangleq\left\lfloor\frac{q-1}{2}\right\rfloor\ .
\]
Here $s$ is the number of raw layers removed, and $\ell$ is the index of the last retained layer below the gap. 

Specifically, we delete the states of raw layers
\[
  \ell+1,\ell+2,\ldots,\ell+s\ .
\]
Raw layers $0,\ldots,\ell$ keep their ranks. Each raw layer $j\in\{\ell+s+1,\ldots,R\}$ is retained at compressed rank $j-s$. The resulting ranks are therefore $0,\ldots,q$. 

In the example of \Cref{fig:211-compressed}, the starting gadget has $c=2$ colors and capacity $r=2$, and we take $m=2$ copies. Hence the raw product has capacity $R=mr=4$. We choose target capacity $q=3$, so $s=R-q=1$ and $\ell=\lfloor(q-1)/2\rfloor=1$. Thus raw layer $2$ is deleted: raw layers $0,1$ remain at compressed ranks $0,1$, while raw layers $3,4$ shift down to compressed ranks $2,3$. Each new shortcut goes from raw layer $1$ to raw layer $3$ and replaces a two-step raw path labeled $h_1a$; the color $h_1$ is hidden and $a$ labels the shortcut. For example, the path $(0,S)\xrightarrow{01}(T,S)\xrightarrow{10}(T,0)$ has $h_1=01$ and $a=10$, and therefore yields the compressed transition $(0,S)\xrightarrow{10}(T,0)$. The other raw path, $(0,S)\xrightarrow{10}(0,0)\xrightarrow{01}(T,0)$, has $h_1=10$ and $a=01$, and therefore yields the parallel compressed transition labeled $01$.

Generally, we keep every raw transition whose endpoints both lie in the retained lower layers or both lie in the retained upper layers, with the same label; only the layer numbers of upper states change under the shift by $s$. These retained raw transitions do not cross the deleted band of layers. To replace paths that do cross it, add a transition $u\xrightarrow{a}v$ from raw layer $\ell$ to raw layer $\ell+s+1$ whenever the raw product has a path from $u$ to $v$ labeled
\[
  h_1h_2\cdots h_s a
\]
for distinct colors $h_1,\ldots,h_s,a$.  Only $a$ is read by the compressed automaton.

The proof of the following is given in \Cref{app:compress-proof}.
\begin{lemma}[Compose-and-compress]\label{lem:compress}
The compressed automaton is a symmetric partial gadget of capacity $q$.  Its state polynomial is
\begin{equation}\label{eq:compressed-state}
  A_{\mathrm{comp}}(z)
  =\sum_{j=0}^{\ell}a_jz^j
   +\sum_{j=\ell+s+1}^{R}a_jz^{j-s}\ ,
\end{equation}
and its certified-set polynomial is
\begin{equation}\label{eq:compressed-good}
  B_{\mathrm{comp}}(z)=\sum_{j=0}^{q}b_jz^j\ .
\end{equation}
\end{lemma}

The next section applies this lemma with a growing number of Witt gadgets.  The section after that applies it once to eleven Witt complete gadgets and then feeds the resulting partial gadget back into the growing construction.

\section{First use: a growing compressed product}\label{sec:amplification}

To obtain an improved NFA construction from CaC, we first build an automaton for an intermediate length $q$.
The reason is the transition count: applying CaC directly at length $k$ gives the desired exponential constant for the number of states, but leaves a transition bound with a factor $2^{O(k)}$ whose hidden constant is too large.

Instead, we choose $q\to\infty$ with $q=o(k)$, construct an inner NFA for length $q$, and combine about $k/q$ inner copies in an outer product.
Each outer transition updates only one inner copy, so the local factor $2^{O(q)}$ occurs only once and is $2^{o(k)}$.
The overhead from the outer hash family is also subexponential in $k$.

The construction therefore has two scales.
At the inner scale, we take a growing number of copies of one fixed partial gadget and apply compose-and-compress to their raw product, obtaining an NFA for repetition-free words of length at most $q$.
At the outer scale, the original $k$ letters are divided among blocks of size $q$, and the inner NFAs are combined by an outer hashing step to recognize $\distlang{k}{n}$.
\Cref{fig:growing-construction-hierarchy} illustrates these two scales: each outer branch contains $T$ intermediate blocks, and each intermediate block is built from $N$ copies of the fixed finite gadget.

\begin{figure}[t]
  \centering
  \includegraphics[width=\textwidth]{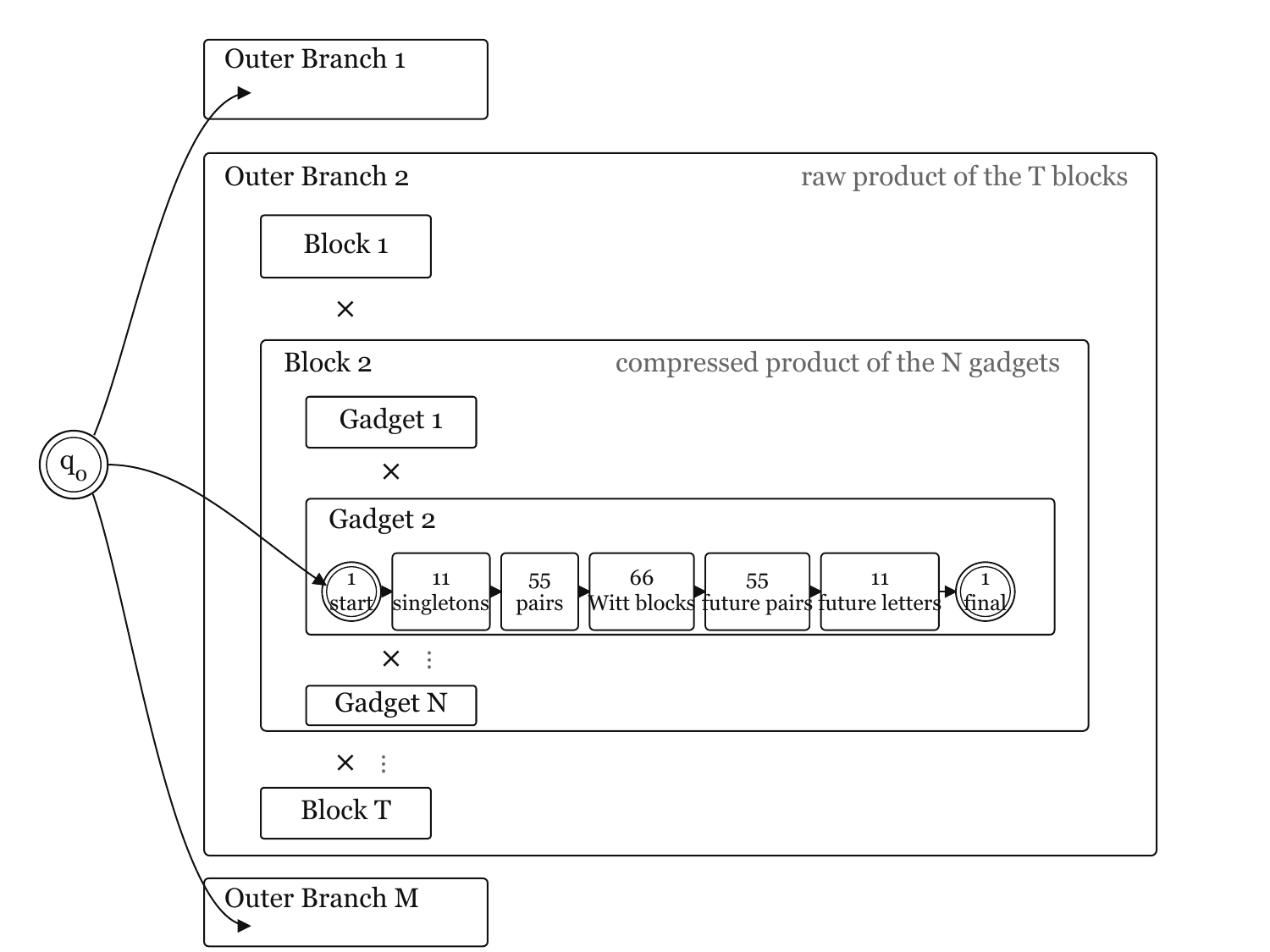}
  \caption{The hierarchy of the growing construction. The common accepting initial state has branches corresponding to the outer hash functions. Each outer branch contains $T$ intermediate blocks, and each intermediate block contains $N$ copies of \mbox{the fixed finite gadget.}}
  \label{fig:growing-construction-hierarchy}
\end{figure}

Throughout this section, let $G$ be a symmetric partial gadget of capacity $r\ge1$ over a $c$-element color set, and write
\[
  A(z)=A_G(z),
  \qquad
  B(z)=B_G(z)=\sum_{j=0}^r B_jz^j\ .
\]

Here is the optimization problem before we introduce its parameters. A branch formed from $N$ copies has one state for each retained product state, so the two retained tails of $A(z)^N$ determine its cost. For a fixed repetition-free $q$-element input set, the probability that the branch certifies the set is proportional to $[z^q]B(z)^N$, the number of certified product sets of size $q$. A larger coefficient reduces the number of required hash branches. Thus we must balance the states per branch against the number of branches. We use the auxiliary variables $x$ and $y$ to control these two costs.

First, we choose a rational number $x\in(0,1)$, and put $a_j\triangleq[z^j]A(z)^N$ and $Y\triangleq\lfloor q/2\rfloor$. The raw product has capacity $Nr$. Central compression to capacity $q$ retains the raw ranks $0,\ldots,\lceil q/2\rceil-1$ and $Nr-q+\lceil q/2\rceil,\ldots,Nr$. Because $G$ is symmetric, $A(z)^N$ is symmetric, and therefore
\[
\sum_{j=Nr-q+\lceil q/2\rceil}^{Nr}a_j=\sum_{j=0}^{q-\lceil q/2\rceil}a_j=\sum_{j=0}^{Y}a_j\ ,
\]
while
\[
\sum_{j=0}^{\lceil q/2\rceil-1}a_j\le\sum_{j=0}^{Y}a_j\ .
\]
Thus each retained tail contains at most $\sum_{j=0}^{Y}a_j$ states. When $q$ is odd the two tails have the same number of layers; when $q$ is even the upper tail has one more layer and contains the central layer.

The coefficients $a_j$ are nonnegative, and $x^j\ge x^Y$ for every $0\le j\le Y$. Hence
\[
x^Y\sum_{j=0}^{Y} a_j\le\sum_{j=0}^{Y} a_jx^j\le A(x)^N\ ,
\]
and therefore
\[
\sum_{j=0}^{Y}[z^j]A(z)^N\le x^{-Y}A(x)^N\ .
\]

Second, a positive rational number $y$ selects the load histogram used to lower-bound $[z^q]B(z)^N$. The next subsection defines that histogram and proves the bound.

The final construction multiplies the number of states per branch by the number of branches, so their base-two logarithms add. We later choose $x$ and $y$ jointly to minimize that sum.

\subsection{A load distribution chosen by one parameter}

The fixed-histogram analysis chose an integral histogram by hand.  We now need the same counting idea for a number of copies that grows with an intermediate length $q$.  It is convenient to describe the histogram by a probability vector rather than by a fixed denominator.

Choose a positive rational number $y$.  Define
\begin{equation}\label{eq:tilt}
  p_j\triangleq\frac{B_jy^j}{B(y)},
  \qquad
  \mu\triangleq\sum_{j=0}^r jp_j
      =\frac{yB'(y)}{B(y)},
  \qquad
  \lambda\triangleq\frac1\mu\ .
\end{equation}
For load $j$, the factor $B_j$ counts the certified local sets of size $j$, while $y^j$ controls how strongly that load is favored. Dividing the weights $B_jy^j$ by their sum $B(y)$ produces the probability vector $(p_0,\ldots,p_r)$. These definitions have the following concrete meanings.
\begin{itemize}[leftmargin=2.2em]
  \item Since $B(y)=\sum_jB_jy^j$, the values $p_j$ are nonnegative and sum to one.
  \item We intend a fraction $p_j$ of the partial-gadget copies to have load $j$.
  \item The average number of colors per copy is $\mu$.
  \item The reciprocal $\lambda=1/\mu$ is the intended number of copies per visible input color.
\end{itemize}

Suppose we take $N$ copies.  We want exactly $Np_j$ copies to have load $j$.
The total number of visible colors is then
\[
  \sum_j jNp_j=N\mu\ .
\]
We therefore write
\[
  q=N\mu,
  \qquad
  N=\lambda q\ .
\]
All numbers in \eqref{eq:tilt} are rational. Hence there is a fixed positive integer $M$ such that, whenever $N$ is a multiple of $M$, every $Np_j$ and $q=N\mu$ is an integer. Writing $N=tM$ gives $q=t(M\mu)$, and $M\mu$ is a positive integer by the choice of $M$. We call these values of $q$ \emph{allowed}. This is only a divisibility condition. Because the allowed values form an arithmetic progression with fixed step $M\mu$, the least allowed $q\ge\sqrt{k}$ satisfies $q=\sqrt{k}+O(1)$; in particular, $q=o(k)$, which is the property used in the final bounds.

The choice $p_j\propto B_jy^j$ is useful because it makes the entropy calculation collapse to a two-term expression in the following lemma.

\begin{lemma}[One explicit coefficient type]\label{lem:coefficient}
For allowed values of $N$ and $q=N\mu$,
\[
  [z^q]B(z)^N
  \ge
  2^{N(\log_2B(y)-\mu\log_2y)-O(\log N)}\ .
  \]
The constant in the error term depends only on the fixed polynomial $B$.
\end{lemma}

\begin{proof}
We repeat the fixed-histogram calculation from \Cref{sec:fixed-profile}, now with $Np_j$ copies of load $j$.

\paragraph{Step 1: keep one histogram}
Write $B(z)^N=\prod_{i=1}^N B(z)$, with the $i$th occurrence of $B(z)$ representing gadget copy $i$. Choosing $B_jz^j$ from that occurrence means that copy $i$ has load $j$. We retain the terms in which exactly $Np_j$ copies have load $j$, for every $j$.  Since
\[
  \sum_jNp_j=N,
  \qquad
  \sum_jjNp_j=N\mu=q\ ,
\]
this is a feasible histogram for the coefficient of $z^q$.  It contributes
\begin{equation}\label{eq:one-type}
  [z^q]B(z)^N
  \ge
  \binom{N}{Np_0,Np_1,\ldots,Np_r}
  \prod_{j=0}^r B_j^{Np_j}\ .
\end{equation}

\paragraph{Step 2: count where the loads occur}
The multinomial coefficient chooses which copies have load $0$, which have load $1$, and so on.  Applying Stirling's formula to its constantly many factorials gives
\[
  \log_2\binom{N}{Np_0,\ldots,Np_r}
  =N\left(-\sum_{j=0}^rp_j\log_2p_j\right)+O(\log N)\ .
\]
The derivation is identical to \eqref{eq:multinomial-entropy}, with $tD$ replaced by $N$.

\paragraph{Step 3: count the local certified sets}
Each of the $Np_j$ copies of load $j$ has $B_j$ choices for its certified local set.  Therefore
\[
  \log_2\left(\prod_{j=0}^rB_j^{Np_j}\right)
  =N\sum_{j=0}^rp_j\log_2B_j\ .
\]
Adding Steps 2 and 3, the logarithm of the right-hand side of \eqref{eq:one-type} is
\[
  N\left(
    -\sum_jp_j\log_2p_j
    +\sum_jp_j\log_2B_j
  \right)+O(\log N)\ .
\]

\paragraph{Step 4: use the special form of $p_j$}
By \eqref{eq:tilt},
\[
  \log_2p_j
  =\log_2B_j+j\log_2y-\log_2B(y)\ .
\]
Substitute this identity into the preceding entropy expression:
\begin{align*}
  &-\sum_jp_j\log_2p_j+\sum_jp_j\log_2B_j\\
  &\quad=
  -\sum_jp_j\bigl(\log_2B_j+j\log_2y-\log_2B(y)\bigr)
  +\sum_jp_j\log_2B_j\\
  &\quad=
  \log_2B(y)\sum_jp_j
  -\log_2y\sum_jjp_j\\
  &\quad=
  \log_2B(y)-\mu\log_2y\ .
\end{align*}
Here we used $\sum_jp_j=1$ and $\sum_jjp_j=\mu$.  Substituting this value into \eqref{eq:one-type} proves the lemma.
\end{proof}

\subsection{Building an NFA for one intermediate block}

We next use $N=\lambda q$ copies of the partial gadget to recognize all repetition-free words of length at most $q$.  Two costs must be controlled:
\begin{enumerate}[label=(\alph*),leftmargin=2.5em]
  \item the number of states retained after compose-and-compress; and
  \item the number of hash functions needed to cover all $q$-element input sets.
\end{enumerate}
The rational parameter $x\in(0,1)$ controls the first cost, while $y$ controls the load distribution and hence the second cost.
Define
\begin{equation}\label{eq:Rxy}
  R(x,y)\triangleq\lambda\log_2A(x)-\frac12\log_2x=\frac{B(y)}{yB'(y)}\log_2A(x)-\frac12\log_2x\ ,
\end{equation}

\begin{equation}\label{eq:Hy}
  H(y)\triangleq\max\left\{0,
  \log_2(c\lambda\e)+\log_2y-\lambda\log_2B(y)
  \right\}\ ,
\end{equation}
and
\begin{equation}\label{eq:Exy}
  E(x,y)\triangleq R(x,y)+H(y)\ .
\end{equation}
The meaning of these quantities will emerge from the proof: one compressed branch has at most $2^{(R(x,y)+o(1))q}$ states, and at most $2^{(H(y)+o(1))q}$ hash branches are needed. Equivalently, $R(x,y)$ and $H(y)$ are the two base-two exponents per visible input color.

\begin{proposition}[Inner NFA]\label{prop:inner}
For every sufficiently large allowed value $q=N\mu$ and every alphabet size $L\ge q$, there is an acyclic NFA $C_{q,L}$ recognizing $\distlang{q}{L}$ with
\begin{equation}\label{eq:inner-states}
  |Q(C_{q,L})|
  \le 2^{(E(x,y)+o(1))q}L^{O(1)}
\end{equation}
and
\begin{equation}\label{eq:inner-transitions}
  |\Delta(C_{q,L})|
  \le 2^{O(q)}L^{O(1)}\ .
\end{equation}
\end{proposition}

\begin{proof}
The construction has five steps.

\paragraph{Step 1: compose $N$ partial-gadget copies and remove the middle}
Take $N=\lambda q$ copies of $G$ on disjoint color sets.  Their raw capacity is
\[
  R_{\mathrm{raw}}\triangleq rN\ .
\]
By \Cref{lem:partial-product}, the product certified-set polynomial is $B(z)^N$.
Properties (P1) and (P4) imply $B_0,B_r>0$, so the tilted distribution in \eqref{eq:tilt} has $0<\mu<r$.
Hence $q=N\mu<Nr=R_{\mathrm{raw}}$, and \Cref{lem:compress} applies with target capacity $q$.
Thus the deleted width and the last retained lower rank are
\[
  s\triangleq R_{\mathrm{raw}}-q,
  \qquad
  \ell\triangleq\left\lfloor\frac{q-1}{2}\right\rfloor\ .
\]
By \Cref{lem:compress}, the compressed product is a partial gadget of capacity $q$.
Its certified sets of size $q$ are exactly those counted by $[z^q]B(z)^N$.

\paragraph{Step 2: count the retained states}
Write
\[
  A(z)^N=\sum_{j=0}^{R_{\mathrm{raw}}}a_jz^j\ .
\]
The coefficient $a_j$ is the number of raw product states at rank $j$.
Compose-and-compress retains the lower ranks $0,\ldots,\ell$ and the upper ranks $\ell+s+1,\ldots,R_{\mathrm{raw}}$.  Hence the number of retained states is
\[
  \sum_{j=0}^{\ell}a_j
  +\sum_{j=\ell+s+1}^{R_{\mathrm{raw}}}a_j\ .
\]
Because the partial gadget $G$ is symmetric, $A(z)^N$ is symmetric, so $a_j=a_{R_{\mathrm{raw}}-j}$.  Reflecting the upper tail around rank $R_{\mathrm{raw}}/2$ changes its index range to $0,\ldots,q-\ell-1$.  Therefore
\[
  \sum_{j=0}^{\ell}a_j
  +\sum_{j=\ell+s+1}^{R_{\mathrm{raw}}}a_j
  =\sum_{j=0}^{\ell}a_j
   +\sum_{j=0}^{q-\ell-1}a_j
  \le2\sum_{j=0}^{\lfloor q/2\rfloor}a_j\ .
\]

We now bound this lower tail by evaluating the polynomial at $x<1$.  For every $j\le\lfloor q/2\rfloor$,
\[
  1\le x^{j-\lfloor q/2\rfloor}\ ,
\]
so
\begin{align*}
  \sum_{j=0}^{\lfloor q/2\rfloor}a_j
  &\le x^{-\lfloor q/2\rfloor}
       \sum_{j=0}^{\lfloor q/2\rfloor}a_jx^j\\
  &\le x^{-\lfloor q/2\rfloor}A(x)^N\ .
\end{align*}
Let $Q_{\mathrm{comp}}$ denote the state set of the compressed product. Its number of states is at most
\[
  |Q_{\mathrm{comp}}|\le 2x^{-\lfloor q/2\rfloor}A(x)^N\ .
\]
  Taking base-two logarithms and using $N=\lambda q$ gives
\[
  \log_2|Q_{\mathrm{comp}}|
  \le
  q\left(\lambda\log_2A(x)-\frac12\log_2x\right)+o(q)
  =(R(x,y)+o(1))q\ .
\]
Thus, the number of states in one compressed product branch is at most
\[
  2^{(R(x,y)+o(1))q}\ .
  \]

\paragraph{Step 3: use one hash's success probability to bound the number of inner hash branches}
Hash the input alphabet $[L]$ into the $cN$ product colors.  Fix a $q$-element input set $S\subseteq[L]$ and choose
\[
  h:[L]\longrightarrow[N]\times C
\]
uniformly at random.
The hash succeeds when its image is one of the certified $q$-element product sets.
Let $\rho$ denote the probability, over the random choice of $h$, that $h$ succeeds on $S$.
This probability determines how many hash branches the final nondeterministic union needs.
The covering argument of \Cref{lem:cover} uses
\[
  O\!\left(\rho^{-1}q\log(2L)\right)
\]
hashes to cover every $q$-element subset of $[L]$.
Because $q\le L$, the factor $q\log(2L)$ is absorbed into $L^{O(1)}$.
After separating this allowed polynomial dependence on $L$, it remains to bound $\rho^{-1}$ as a function of $q$.
Our goal in this step is therefore to prove
\[
  \rho^{-1}\le 2^{(H(y)+o(1))q}\ ,
\]
or equivalently,
\[
  \frac1q\log_2\rho^{-1}\le H(y)+o(1)\ .
\]

There are $[z^q]B(z)^N$ choices for such a certified set.  Once it is chosen, the $q$ labeled elements of $S$ may be bijected to its cells in $q!$ ways.
There are $(cN)^q$ possible restrictions of $h$ to $S$.  Therefore
\begin{equation}\label{eq:inner-success}
  \rho
  =\frac{q!}{(cN)^q}[z^q]B(z)^N\ .
\end{equation}

To prove this target bound, we estimate $\rho^{-1}$.  \Cref{lem:coefficient} gives
\[
  \log_2[z^q]B(z)^N
  \ge N\bigl(\log_2B(y)-\mu\log_2y\bigr)-O(\log N)\ .
\]
Since $N=\lambda q$ and $N\mu=q$, this becomes
\begin{align}\label{eq:zqBzN_lb}
  \log_2[z^q]B(z)^N
  \ge
  q\bigl(\lambda\log_2B(y)-\log_2y\bigr)-o(q)\ .
\end{align}
Also, Stirling's formula and $cN=c\lambda q$ give
\[
  \log_2\frac{(cN)^q}{q!}
  =q\log_2(c\lambda\e)+o(q)\ .
  \]
By \eqref{eq:inner-success},
\[
  \log_2\rho^{-1}
  =\log_2\frac{(cN)^q}{q!}
   -\log_2[z^q]B(z)^N\ .
\]
Applying \eqref{eq:zqBzN_lb} gives
\begin{align*}
  \log_2\rho^{-1}
  &\le q\log_2(c\lambda\e)+o(q)\\
  &\quad-\left[q\bigl(\lambda\log_2B(y)-\log_2y\bigr)-o(q)\right]\\
  &=q\left(\log_2(c\lambda\e)+\log_2y-\lambda\log_2B(y)\right)+o(q)\ .
\end{align*}

Dividing by $q$ and using $o(q)/q=o(1)$ gives
\[
  \frac1q\log_2\rho^{-1}
  \le
  \log_2(c\lambda\e)+\log_2y-\lambda\log_2B(y)+o(1)\ .
\]
Finally, \eqref{eq:Hy} defines $H(y)$ as the maximum of zero and the three-term expression on the right.
That expression is therefore at most $H(y)$.
Hence
\[
  \frac1q\log_2\rho^{-1}\le H(y)+o(1)\ .
\]
Equivalently, $\rho^{-1}\le 2^{(H(y)+o(1))q}$.
By \Cref{lem:cover}, the number of required hash branches is therefore at most $2^{(H(y)+o(1))q}L^{O(1)}$, where the factor $q\log(2L)$ is absorbed into $L^{O(1)}$.

\paragraph{Step 4: cover every input set and add the exponents}
Choose
\[
  O\!\left(\rho^{-1}q\log(2L)\right)
\]
independent hashes.  The same union-bound argument as in \Cref{lem:cover} shows that some such family covers every $q$-set.  Since the certified family is downward closed, the same family covers every smaller set: extend it to a $q$-set, use a successful hash on the extension, and then restrict.

For each chosen hash $h:[L]\to[N]\times C$, take a separate copy of the compressed product. If the compressed product contains a transition $u\xrightarrow{(i,c)}v$, then in the copy associated with $h$ add the transition $u\xrightarrow{a}v$ for every input symbol $a\in[L]$ satisfying $h(a)=(i,c)$. Thus reading $a$ in this copy has exactly the same effect as reading its product color $h(a)$ in the compressed product.  The nondeterministic union of these copies, one for each chosen hash, recognizes $\distlang{q}{L}$.  One branch has state exponent $R(x,y)$ by Step 2, and the number of branches has exponent $H(y)$ by Step 3.  Therefore the union has exponent
\[
  R(x,y)+H(y)=E(x,y)\ ,
\]
which proves \eqref{eq:inner-states}.

Correctness is immediate. If an input symbol repeats, both occurrences receive the same product color; because every compressed path is repetition-free, every branch rejects. If the input word is repetition-free, choose a hash that certifies its underlying set. Property (P3) ensures that the corresponding branch accepts every ordering of that certified set, including the input word.

\paragraph{Step 5: ordinary transitions}
For the final theorem we also need a bound on the number of transitions.
One compressed product has no more than the $A(1)^N=2^{O(q)}$ states of its raw product.  Its color alphabet has $cN=O(q)$ elements.  Even the crude bound that allows every state-color pair to connect to every state gives only $2^{O(q)}$ transitions.  For one fixed hash, the construction in Step~4 replaces each compressed-product transition with at most $L$ transitions over $[L]$, so it increases the transition count by at most a factor $L$. The number of hashes is $2^{O(q)}L^{O(1)}$.  Hence
\[
  |\Delta(C_{q,L})|\le2^{O(q)}L^{O(1)}\ ,
\]
which is \eqref{eq:inner-transitions}.
\end{proof}

\subsection{From intermediate blocks to the full input}

The inner NFA recognizes $\distlang{q}{L}$.  Its state bound has the desired constant $E(x,y)$, but its transition bound is only $2^{O(q)}L^{O(1)}$, with an unspecified constant in the exponent.  We prevent that local transition constant from being multiplied by $k/q$ by using an outer product in which one transition updates only one component.

Choose an allowed value $q=q(k)$ such that
\[
  q\longrightarrow\infty,
  \qquad
  q=o(k)\ .
\]
Because the allowed values of $q$ described after \eqref{eq:tilt} are multiples of the fixed integer $M\mu$, we may take the least allowed value at least $\sqrt{k}$.  Let $K$ be the least multiple of $q$ with $K\ge k$, and define
\[
  T\triangleq K/q,
  \qquad
  L\triangleq q^3\ .
\]
These parameters have simple roles:
\begin{itemize}[leftmargin=2.2em]
  \item $K$ is a padded target length divisible by $q$;
  \item $T$ is the number of outer blocks;
  \item every outer block will contain exactly $q$ symbols; and
  \item $L=q^3$ is large enough to give distinct local names to the symbols inside every block at subexponential cost.
\end{itemize}
Since $0\le K-k<q=o(k)$, we have $K=k+o(k)$.

The first-coordinate requirement below is exactly a splitter, while the second-coordinate requirement assigns distinct local names to the elements mapped to each first-coordinate value. The precise deterministic composition used here is proved in \Cref{app:derandomization} from the splitter and perfect-hash constructions of Naor, Schulman, and Srinivasan~\cite[Sections~2.2 and~4.4, pp.~183, 186--187]{NaorSchulmanSrinivasan1995}.

\begin{lemma}[Balanced outer hashes]\label{lem:outer-hash}
For every domain size $u\ge K$, there is a family of $2^{o(K)}u^{O(1)}$ functions
\[
  g:[u]\longrightarrow[T]\times[L]
\]
such that, for every $K$-element subset of $[u]$, some $g$ sends exactly $q$ elements to each outer block and is injective on the second coordinates inside each block.
\end{lemma}

\begin{proof}
\Cref{app:outer-hashes} gives the probability calculation and then the stronger globally explicit construction used by the deterministic algorithm.
\end{proof}

\begin{theorem}[Amplification theorem]\label{thm:amplification}
Let $G$ be a fixed symmetric partial gadget of capacity $r\ge1$ over a $c$-element color set, with polynomials $A$ and $B$.
Then, for every rational $x\in(0,1)$, positive rational $y$, and $1\le k\le n$, there exists an acyclic NFA recognizing $\distlang{k}{n}$ with
\[
  |Q|+|\Delta|
  \le2^{(E(x,y)+o(1))k}n^{O(1)}\ .
\]
\end{theorem}

\begin{proof}
The proof has five steps.

\paragraph{Step 1: pad the target length}
Temporarily add $K-k$ fresh symbols to $[n]$.  The padded alphabet has size
\[
  \npad\triangleq n+K-k\ .
\]
We construct an automaton for $\distlang{K}{\npad}$ and later remove the fresh symbols.

\paragraph{Step 2: build one branch for one outer hash}
Construct the inner NFA $C_{q,L}$ from \Cref{prop:inner}.  Fix one function
\[
  g:[\npad]\longrightarrow[T]\times[L]
\]
from \Cref{lem:outer-hash}.  Take the raw product of $T$ copies of $C_{q,L}$. Write $g(a)=(g_1(a),g_2(a))$.
When the product automata reads a symbol $a\in[\npad]$, the value $g_1(a)\in[T]$ selects one inner copy, and $g_2(a)\in[L]$ is the {local symbol read by that copy.}

\paragraph{Step 3: proving the correctness of the union of branches}
Suppose an input symbol $a$ occurs twice.  Both occurrences have the same value $g(a)$.  The same inner copy therefore sees the same local symbol twice, and that copy rejects.  Thus every branch rejects every word with a repeated symbol.

Now suppose the input word is repetition-free and has length at most $K$. Let $S$ be its set of symbols.
Because $\npad=n+K-k\ge K$, extend $S$ to a $K$-element set $S'\subseteq[\npad]$, and choose an outer hash $g$ that is successful on $S'$.
By the definition of success in \Cref{lem:outer-hash}, for every $i\in[T]$ there are exactly $q$ elements $a\in S'$ with $g_1(a)=i$, and their values $g_2(a)$ are pairwise distinct.
Passing from $S'$ to $S$ only deletes elements.
Hence each inner copy receives at most $q$ symbols, and their local names remain pairwise distinct.
Each copy of $C_{q,L}$ therefore accepts its local word, so the branch for $g$ accepts the original word.
Taking the nondeterministic union over the outer hashes recognizes $\distlang{K}{\npad}$.

\paragraph{Step 4: count states}
Let
\[
  S_q=|Q(C_{q,L})|\ .
\]
One outer product branch has at most $S_q^T$ states.  By \eqref{eq:inner-states},
\[
  S_q
  \le2^{(E(x,y)+o(1))q}L^{O(1)}\ .
\]
Raise this bound to the power $T=K/q$:
\[
  S_q^T
  \le
  2^{(E(x,y)+o(1))K}L^{O(K/q)}\ .
\]
Since $L=q^3$,
\[
  L^{O(K/q)}
  =q^{O(K/q)}
  =2^{O((K/q)\log q)}
  =2^{o(K)}\ .
\]
Thus, the number of states in one branch is
\[
  2^{(E(x,y)+o(1))K}\ .
\]
By \Cref{lem:outer-hash}, the outer family has at most $2^{o(K)}\npad^{O(1)}$ functions.  Multiplying this by the number of states in one branch gives the explicit bound
\[
  |Q|\le 2^{(E(x,y)+o(1))K}\npad^{O(1)}\ .
\]

\paragraph{Step 5: count ordinary transitions and remove the padding}
Let
\[
  D_q=|\Delta(C_{q,L})|\ .
\]
A transition of the outer product updates one of the $T$ components.  Choose the other $T-1$ local states, choose one local transition in the updated component, and choose the original input symbol that triggers it.  This gives the bound
\[
  \npad D_qS_q^{T-1}
\]
for one branch.  By \eqref{eq:inner-transitions}, $D_q=2^{O(q)}L^{O(1)}$.  Since $S_q\ge1$, the number of transitions in one branch is therefore at most
\begin{align*}
  \npad D_qS_q^{T-1}
  &\le \npad D_qS_q^T\\
  &\le 2^{(E(x,y)+o(1))K}\npad^{O(1)}\ ,
\end{align*}
where we used $q=o(K)$ and $L=q^3$ to absorb $D_q$ into the $o(K)$ and polynomial terms.  Multiplying by the at most $2^{o(K)}\npad^{O(1)}$ outer hashes from \Cref{lem:outer-hash} gives
\[
  |\Delta|\le 2^{(E(x,y)+o(1))K}\npad^{O(1)}\ .
\]

Finally, delete every transition labeled by one of the $K-k$ fresh symbols and retain only layers $0,\ldots,k$. Every component automaton is layered, and every ordinary transition advances by one layer. Raw products, nondeterministic unions formed by merging initial states, and the deletion of transitions or terminal layers therefore preserve acyclicity. The resulting automaton is acyclic and recognizes $\distlang{k}{n}$. We have $K=k+o(k)$ and $\npad=O(n)$, so the claimed bound follows.
\end{proof}

\subsection{The sharper Witt bound}

Regard the Witt complete gadget as a partial gadget by certifying every color set of size at most six. Choose
\[
  x_0=\frac{173}{250}=0.692,
  \qquad
  y_0=\frac{1547}{1000}=1.547\ .
\]
Substitution into \eqref{eq:tilt} and \eqref{eq:Rxy}--\eqref{eq:Exy} gives
\[
\begin{array}{c|c}
\text{quantity}&\text{value}\\
\midrule
\mu(y_0)&5.2105167\ldots\\
\lambda&0.1919195\ldots\\
H(y_0)&0.5257292\ldots\\
R(x_0,y_0)&1.4469540\ldots\\
E(x_0,y_0)&1.9726832\ldots
\end{array}
\]
and therefore
\begin{equation}\label{eq:witt-alone}
  2^{E(x_0,y_0)}=3.9249744\ldots<3.925\ .
\end{equation}

By folding the $o(1)$ factor into the numerical slack of~\Cref{eq:witt-alone}, we conclude the following.
\begin{corollary}[Sharper Witt bound]\label{cor:witt-sharp}
There is an acyclic NFA recognizing $\distlang{k}{n}$ of size
\[
  3.925^k n^{O(1)}\ .
\]
\end{corollary}

The finite automaton in the fixed-histogram and sharper Witt bounds is identical. The improvement from \eqref{eq:witt-fixed} to \eqref{eq:witt-alone} is the first use of compose-and-compress: the number of Witt-gadget copies grows with $q$, the middle band is removed, and only the two tails are counted. The next section uses the same operation at finite scale to change the partial gadget itself.

\section{Second use: an eleven-Witt partial gadget}\label{sec:eleven-witt}

The first use compressed a growing product while leaving the finite Witt complete gadget unchanged.  We now use the same operation once at finite scale to construct a better partial gadget, and then feed that partial gadget back into the growing construction of \Cref{sec:amplification}.
The eleven-copy construction below is the finite composition used in the main theorem.

Take eleven Witt complete gadgets on disjoint eleven-color sets.  The raw product has $121$ colors, capacity $R=66$, and state polynomial $A_W(z)^{11}$.
Set the target capacity to
\[
  q_{11}=61\ .
\]
Then $s= R-q_{11}=5$ and
\[
  \ell=\left\lfloor\frac{q_{11}-1}{2}\right\rfloor=30\ ,
\]
so CaC deletes raw ranks $31,\ldots,35$.

\begin{figure}[H]
\centering
\includegraphics[width=0.94\textwidth]{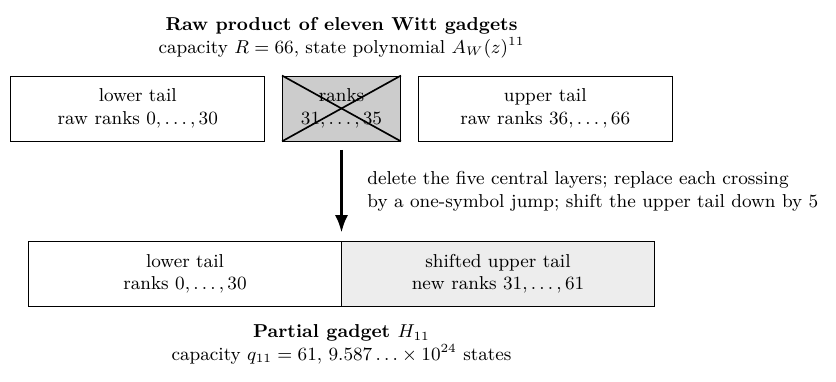}
\caption{The one finite composition used in the main theorem. Eleven Witt complete gadgets are composed, five central layers of their raw product are deleted, and the result is a capacity-$61$ partial gadget.}
\label{fig:eleven-witt}
\end{figure}

Call the resulting partial gadget $H_{11}$.  If
\[
  A_W(z)^{11}=\sum_{j=0}^{66}a_jz^j,
  \qquad
  B_W(z)^{11}=\sum_{j=0}^{66}b_jz^j\ ,
\]
then \Cref{lem:compress} gives
\begin{equation}\label{eq:Aeleven}
  A_{11}(z)
  =\sum_{j=0}^{30}a_jz^j
   +\sum_{j=36}^{66}a_jz^{j-5}\ ,
\end{equation}
\begin{equation}\label{eq:Beleven}
  B_{11}(z)=\sum_{j=0}^{61}b_jz^j\ .
\end{equation}
The polynomial $A_{11}$ is symmetric.  The raw product has
\[
  200^{11}=20{,}480{,}000{,}000{,}000{,}000{,}000{,}000{,}000
\]
states.  The five deleted layers contain
\[
  \sum_{j=31}^{35}[z^j]A_W(z)^{11}
  =10{,}892{,}645{,}599{,}457{,}631{,}372{,}405{,}834
\]
states, so
\[
  A_{11}(1)=9{,}587{,}354{,}400{,}542{,}368{,}627{,}594{,}166\ .
  \]
The exact coefficients are listed in \Cref{app:eleven-coefficients}.

\subsection{The final numerical bound}

Choose the rational values
\[
  x=\frac{153}{200}=0.765,
  \qquad
  y=\frac{81}{50}=1.62\ .
\]
Direct evaluation gives (see \Cref{app:final-bound-verification} for verification)
\[
  A_{11}(x)=3.481206546463651806859\ldots\times10^{21}\ ,
\]
\[
  B_{11}(y)=2.439760214853457547554\ldots\times10^{46}\ ,
\]
and
\[
  \mu(y)=\frac{yB_{11}'(y)}{B_{11}(y)}
  =57.2020375407\ldots,
  \qquad
  \lambda=\frac 1\mu=0.01748189475\ldots\ .
\]
Here $c=121$, because $H_{11}$ uses eleven disjoint eleven-color palettes.
Substituting the displayed values into the two terms of \eqref{eq:Rxy}--\eqref{eq:Hy} gives
\[
\begin{aligned}
  H(y)
    &=\log_2(121\lambda\e)+\log_2y-\lambda\log_2B_{11}(y)
      <0.525678\ ,\\
  R(x,y)
    &=\lambda \log_2A_{11}(x)-\frac12\log_2x
      <1.444240\ .
\end{aligned}
\]
Consequently,
\[
  E(x,y)<1.969918
\]
and
\begin{equation}\label{eq:final-number}
  2^{E(x,y)}<3.917459<3.918\ .
\end{equation}
Finally, the slack in \eqref{eq:final-number} absorbs the $o(1)$ term. Applying \Cref{thm:amplification}, together with the deterministic construction in \Cref{app:derandomization}, proves \Cref{thm:main}.

\section{Discussion}

The paper develops a gadget-amplification framework and improves it at three successive levels. Full-product amplification with the Witt gadget gives an $O^*(3.967^k)$-size NFA. Applying compose-and-compress at the growing scale improves the bound to $O^*(3.925^k)$. Applying finite compression first to obtain $H_{11}$ and then using growing amplification gives the proved $O^*(3.918^k)$ bound.

The final proved constant is not intrinsic to the amplification argument. One route to a better bound is a better starting complete gadget, obtained for example from a smaller balanced separating family or from a different finite design. A broader possibility is to start directly from a partial gadget whose certified sets have been chosen with later composition in mind, rather than obtaining partiality only by compressing a complete gadget. Such a gadget may certify fewer local color sets while having a state profile small enough to improve the combined exponent.

Stopping after one finite composition is a choice made for verifiability, not a limitation of the compose-and-compress method. In exploratory computations, starting from the Witt gadget and applying the sequence\[  \langle5,1\rangle,\quad  \langle3,4\rangle,\quad  \langle3,8\rangle,\quad  \langle2,5\rangle,\quad  \langle2,7\rangle,\quad  \langle2,9\rangle\ ,\] where $\langle m,s\rangle$ denotes composing $m$ copies of the current gadget and compressing a central band of width $s$, gave a size of $O^*(3.91322^k)$.
The exact state and certified-set polynomials grow rapidly along this sequence, requiring a computer-assisted computation and verification. We have preferred the slightly weaker bound from one finite composition because every finite input to that calculation can be inspected directly in the manuscript.

The most basic open problem is to close the gap between the known $2^k$ lower bound and the $3.918^k$ upper bound. A separate question is whether our construction can be leveraged to break the $4^k$ barrier for deterministic approximate counting of $k$-paths. For every fixed $\varepsilon>0$, Lokshtanov, Bj{\"o}rklund, Saurabh, and Zehavi give a deterministic $(1\pm\varepsilon)$-approximation algorithm running in $O^*(4^{k+o(k)})$ time~\cite{LokshtanovBjorklundSaurabhZehavi2021}. A direct application of our NFA does not give such an algorithm, because different repetition-free words may have different numbers of accepting runs. It would therefore be interesting to determine whether the construction can be made approximately balanced, so that every repetition-free word has approximately the same known number of accepting runs, while retaining \mbox{fewer than $4^k$ states.}

\clearpage
\appendix
\numberwithin{equation}{section}
\crefalias{section}{appsection}
\crefalias{subsection}{appsection}
\makeatletter
\@addtoreset{equation}{section}
\@addtoreset{figure}{section}
\@addtoreset{table}{section}
\makeatother
\renewcommand{\theequation}{\Alph{section}.\arabic{equation}}
\renewcommand{\thefigure}{\Alph{section}.\arabic{figure}}
\renewcommand{\thetable}{\Alph{section}.\arabic{table}}

\section{How the Witt histogram was selected}\label{app:witt-histogram-selection}

The theorem is valid for every feasible integral template.  We now explain the particular numbers used for the Witt gadget.

Before doing any calculus, let us state exactly what is being optimized.  For fixed $K$ and $N$, \eqref{eq:success} reads
\[
  p_{K,N}
  =\frac{K!}{(cN)^K}[z^K]B(z)^N\ .
\]
The factor $K!/(cN)^K$ is then fixed.  Thus a larger lower bound on $[z^K]B(z)^N$ gives a larger success probability $p_{K,N}$, and \Cref{lem:cover} then needs fewer hash functions.  This is why the coefficient is the quantity we want to make large.

A load distribution $\mathbf p$ used by $N$ copies must have mean
\[
  \sum_j j p_j=\frac{K}{N}=: \mu\ .
\]
For a rational distribution that is repeated as one histogram, \eqref{eq:coefficient-before-rescale} gives the lower bound
\[
  [z^K]B(z)^N
  \ge 2^{N\Psi_B(\mathbf p)-o(N)}\ .
\]
Consequently, when $K/N=\mu$ is fixed, the correct inner optimization problem is to maximize $\Psi_B(\mathbf p)$ over distributions of mean $\mu$.

There is also an outer tradeoff.  The mean $\mu$ is not fixed by the problem: changing it changes
\[
  \lambda=\frac{N}{K}=\frac1\mu\ .
\]
This changes both the hash exponent and the product-state cost $s(G)^N=s(G)^{\lambda K}$ in \eqref{eq:profile-exponent}.  We therefore make the choice in two stages:
\begin{enumerate}[label=(\roman*),leftmargin=2.4em]
  \item for each fixed mean $\mu$, find the distribution that gives the largest coefficient contribution; and
  \item vary $\mu$, or equivalently $\lambda$, to minimize the total NFA exponent.
\end{enumerate}
We now solve the fixed-mean optimization problem.  Let $J=\{j:B_j>0\}$, and consider real distributions $\mathbf p=(p_j)_{j\in J}$ satisfying
\[
  \sum_{j\in J}p_j=1,
  \qquad
  \sum_{j\in J}jp_j=\mu\ .
\]
Multiplying the objective by $\ln 2$ does not change its maximizer, so it is equivalent to maximizing
\[
  \widehat\Psi_B(\mathbf p)
  \triangleq(\ln 2)\Psi_B(\mathbf p)
  =\sum_{j\in J}p_j\ln\frac{B_j}{p_j}\ .
\]
Write $j_{\min}=\min J$ and $j_{\max}=\max J$. The possible mean values are exactly $j_{\min}\le\mu\le j_{\max}$. Fix a target average load with $j_{\min}<\mu<j_{\max}$. We first solve the stationarity equations subject to the two displayed constraints and the explicit conditions $p_j>0$ for every $j\in J$. The relative-entropy argument below will show that the resulting positive distribution is the unique global maximizer. Add multipliers $\alpha$ and $\theta$ for the two displayed constraints. The Lagrangian is
\[
  \mathcal L(\mathbf p,\alpha,\theta)
  =\sum_{j\in J}p_j\ln\frac{B_j}{p_j}
   +\alpha\left(\sum_{j\in J}p_j-1\right)
   +\theta\left(\sum_{j\in J}jp_j-\mu\right)\ .
\]
Differentiating with respect to $p_j$ gives one equation for each load $j$:
\[
  0=\frac{\partial\mathcal L}{\partial p_j}
   =\ln B_j-\ln p_j-1+\alpha+\theta j\ .
\]
Solving this equation for $p_j$ yields
\[
  p_j=B_j\e^{\alpha-1}\e^{\theta j}\ .
\]
The multiplier $\theta$ appears only through the powers $\e^{\theta j}$.
Define
\[
  y\triangleq\e^\theta>0\ .
\]
This single positive parameter indexes the stationary distributions; it is not an additional load, probability, or constraint.  The normalization constraint now determines the remaining constant, because
\[
  1=\sum_{j\in J}p_j
   =\e^{\alpha-1}\sum_{j\in J}B_jy^j
   =\e^{\alpha-1}B(y)\ .
\]
Consequently every stationary distribution has the tilted form
\begin{equation}\label{eq:fixed-tilt}
  p_j(y)=\frac{B_jy^j}{B(y)}
  \qquad (y>0)\ .
\end{equation}

The mean constraint selects the value of $y$.  Direct substitution also gives both identities used below:
\begin{equation}\label{eq:fixed-tilt-identities}
\begin{aligned}
  \mu(y)
  &=\sum_{j\in J}j\,p_j(y)\\
  &=\frac{1}{B(y)}\sum_{j\in J}jB_jy^j
    =\frac{yB'(y)}{B(y)}\\[2pt]
  \Psi_B(\mathbf p(y))
  &=\sum_{j\in J}p_j(y)
      \bigl(\log_2B_j-\log_2p_j(y)\bigr)\\
  &=\sum_{j\in J}p_j(y)
      \bigl(\log_2B(y)-j\log_2y\bigr)\\
  &=\log_2B(y)-\mu(y)\log_2y\ .
\end{aligned}
\end{equation}
To see how the mean changes with $y$, fix a load $j$ and differentiate $p_j(y)=B_jy^j/B(y)$.  Here $j$ and $B_j$ are constants; only $y^j$ and $B(y)$ depend on $y$.  The quotient rule gives
\begin{align*}
  \frac{d p_j(y)}{dy}
  &=B_j\frac{jy^{j-1}B(y)-y^jB'(y)}{B(y)^2}\\
  &=\frac{B_jy^j}{B(y)}
    \left(\frac{j}{y}-\frac{B'(y)}{B(y)}\right)\\
  &=p_j(y)\left(\frac{j}{y}-\frac{B'(y)}{B(y)}\right)\ .
\end{align*}
The identity $\mu(y)=yB'(y)/B(y)$ says that $B'(y)/B(y)=\mu(y)/y$.  Substituting it into the last line yields
\begin{equation}\label{eq:tilted-probability-derivative}
  \frac{d p_j(y)}{dy}
  =\frac{p_j(y)}{y}\bigl(j-\mu(y)\bigr)\ .
\end{equation}
Thus increasing $y$ increases the probability of every load above the current mean and decreases the probability of every load below it.
\par\pagebreak[3]
Since $\mu(y)=\sum_{j\in J}j\,p_j(y)$, differentiating term by term and using \eqref{eq:tilted-probability-derivative} gives
\begin{align*}
  \mu'(y)
  &=\sum_{j\in J}j\,\frac{d p_j(y)}{dy}\\
  &=\frac1y\sum_{j\in J}j p_j(y)\bigl(j-\mu(y)\bigr)\\
  &=\frac1y\left(\sum_{j\in J}j^2p_j(y)
       -\mu(y)\sum_{j\in J}j p_j(y)\right)\\
  &=\frac1y\left(\sum_{j\in J}j^2p_j(y)-\mu(y)^2\right)\\
  &=\frac{\operatorname{Var}_{\mathbf p(y)}(j)}{y}\ .
\end{align*}
Because $y>0$, this derivative is positive unless only one load is available.
Thus in the Witt case, where all loads $0,\ldots,6$ are available, $\mu(y)$ increases continuously from $0$ to $6$.  Therefore, given any target average load $\mu$ with $0<\mu<6$, exactly one $y>0$ solves $\mu(y)=\mu$.  The cases $\mu=0$ and $\mu=6$ arise as the limits $y\to0$ and $y\to\infty$.

The Lagrange equations identify a stationary candidate.  To prove that this candidate is the global maximizer, we compare its objective value with that of every other feasible distribution.  Fix $y>0$, write $\mathbf q=\mathbf p(y)$, and let $\mathbf p$ be any other distribution with the same mean $\mu(y)$.  From \eqref{eq:fixed-tilt},
\[
  \log_2q_j
  =\log_2B_j+j\log_2y-\log_2B(y)\ ,
\]
so the base-two relative entropy can be expanded one term at a time:
\begin{align*}
  D_2(\mathbf p\Vert\mathbf q)
  &\triangleq\sum_{j\in J}p_j\log_2\frac{p_j}{q_j}\\
  &=\sum_{j\in J}p_j\log_2p_j
    -\sum_{j\in J}p_j\log_2B_j\\
  &\quad-\left(\sum_{j\in J}jp_j\right)\log_2y
    +\left(\sum_{j\in J}p_j\right)\log_2B(y)\\
  &=-\Psi_B(\mathbf p)-\mu(y)\log_2y+\log_2B(y)\ .
\end{align*}
The last equality uses both feasibility constraints: $\sum_jp_j=1$ and $\sum_jjp_j=\mu(y)$.  Applying the second identity in \eqref{eq:fixed-tilt-identities} to $\mathbf q$ gives
\[
  \Psi_B(\mathbf q)=\log_2B(y)-\mu(y)\log_2y\ .
\]
Comparing the last two displays yields the exact gap identity
\[
  D_2(\mathbf p\Vert\mathbf q)
  =\Psi_B(\mathbf q)-\Psi_B(\mathbf p)\ .
\]
It remains only to note why the left-hand side is nonnegative.  For each $j$ with $p_j>0$, apply $-\ln x\ge 1-x$ to $x=q_j/p_j$.  Then
\[
  (\ln 2)D_2(\mathbf p\Vert\mathbf q)
  =-\sum_{j:p_j>0}p_j\ln\frac{q_j}{p_j}
  \ge\sum_{j:p_j>0}(p_j-q_j)
  =1-\sum_{j:p_j>0}q_j
  \ge0\ .
\]
The final inequality holds because $\mathbf q$ is a probability distribution, so the sum of its entries over any subset of $J$ is at most one.  Equality in all the preceding inequalities is possible only when $\mathbf p=\mathbf q$.
Therefore
\[
  \Psi_B(\mathbf p)\le\Psi_B(\mathbf q)\ ,
\]
with equality only for $\mathbf p=\mathbf q$.  This proves that the tilted distribution is the unique global maximizer for its mean.

We can now see exactly what varying $y$ does.  Define
\[
  \lambda(y)\triangleq\frac1{\mu(y)}\ .
\]
For this value of $\lambda$, the tilt $\mathbf p(y)$ is the best continuous histogram with mean $1/\lambda(y)$.  Substituting \eqref{eq:fixed-tilt-identities} into \eqref{eq:profile-exponent} turns the continuous search into the one-variable minimization of
\begin{align*}
  H_{\mathrm{raw}}(y)
  &=\log_2(c\lambda(y)\e)\\
  &\quad-
    \lambda(y)\bigl(\log_2B(y)-\mu(y)\log_2y\bigr)\ .
\end{align*}
Here $H_{\mathrm{raw}}(y)$ is the hash rate before truncating it at zero.
The total exponent is therefore
\[
  E_{\mathrm{tilt}}(y)
  =\lambda(y)\log_2s(G)
   +\max\{0,H_{\mathrm{raw}}(y)\}\ .
\]
Increasing $y$ gives the larger loads more weight through the factor $y^j$.
This raises the mean $\mu(y)$ and lowers the number $\lambda(y)=N/K$ of gadget copies per input symbol, but it also changes the coefficient and hash contributions.  The displayed function records this tradeoff completely.

Thus varying $y$ allows us to optimize the continuous histogram for each possible mean in terms of the final NFA exponent.  At the Witt minimum, $H_{\mathrm{raw}}(y)=0.5691\ldots>0$, so the maximum selects $H_{\mathrm{raw}}(y)$.  Substituting $\lambda(y)=1/\mu(y)$ therefore gives
\[
  E_{\mathrm{tilt}}(y)
  =\log_2y+\log_2\frac{c\e}{\mu(y)}
   +\frac{\log_2s(G)-\log_2B(y)}{\mu(y)}\ .
\]
This is an ordinary one-variable function.  A numerical minimization for the Witt values $c=11$, $s(G)=200$, and $B=B_W$ gives
\[
  y=1.942858778\ldots\ .
\]
Equivalently, differentiating the last display with respect to $t=\ln y$ shows that an interior stationary point satisfies
\[
  \frac{dE_{\mathrm{tilt}}}{dt}
  =-\frac{\operatorname{Var}_{\mathbf p(y)}(j)}
          {\mu(y)^2\ln2}
    \left(\mu(y)+\ln\frac{s(G)}{B(y)}\right)=0\ ,
\]
or
\[
  B(y)=s(G)\e^{\mu(y)}\ .
\]
Solving this equation gives the same numerical value.  We round it to $y\approx1.94286$, for which
\[
\begin{array}{c|rrrrrrr}
  j&0&1&2&3&4&5&6\\
\hline
  1000p_j(y)&0.023&0.488&4.742&27.640&107.402&292.133&567.572\ .
\end{array}
\]
Rounding these seven values gives the simple integral template in \Cref{tab:witt-loads}.  The first two entries round to zero, so only five loads need to be displayed.  The proof in the main text does not assume that this rounded template is optimal; it recomputes the final exponent directly from the five displayed integers.

There is also no asymptotic reason to sum every histogram in \eqref{eq:all-histograms}.  For fixed capacity $r$, the number of feasible histograms is at most $(N+1)^{r+1}$, which is polynomial in $N$.  Hence the largest single summand and the full coefficient differ by at most a polynomial factor.  The denominator $1000$ is used only to make the chosen approximation short and transparent.

\section{Derandomization and uniform construction}\label{app:derandomization}

The main proof uses random hash functions only to estimate how many hashes are needed.  This appendix replaces those hashes by deterministic families and shows that the complete NFA can be listed within the asymptotic bound of \Cref{thm:main}.  There are two families to construct:
\begin{enumerate}[label=(\roman*),leftmargin=3em]
  \item the \emph{inner hashes} in \Cref{prop:inner}, which map the alphabet $[L]$ into the colors of a product partial gadget; and
  \item the \emph{outer hashes} in \Cref{lem:outer-hash}, which divide the input symbols evenly among $T$ outer blocks and are injective inside each block.
\end{enumerate}
We use the globally explicit splitter and $k$-restriction constructions of Naor, Schulman, and Srinivasan~\cite[Sections~3 and~4.4, pp.~184--187]{NaorSchulmanSrinivasan1995}.

For reference, the parameters have the same meanings as in the main proof:
\begin{center}
\begin{tabular}{@{}cl@{}}
$k$ & requested maximum word length,\\
$q$ & intermediate block capacity, chosen as $\Theta(\sqrt{k})$,\\
$K$ & the least multiple of $q$ with $K\ge k$,\\
$T=K/q$ & number of outer blocks,\\
$L=q^3$ & alphabet size of one inner NFA.
\end{tabular}
\end{center}
The divisibility adjustment adds $K-k$ fresh symbols to the input alphabet.
We call the resulting alphabet the \emph{padded alphabet} and denote its size by
\begin{equation}\label{eq:npad}
  \npad\triangleq n+K-k\ .
\end{equation}
The symbol $\npad$ is therefore notation, not a second input parameter.  At the end we delete all transitions labeled by the fresh symbols and keep only ranks $0,\ldots,k$.  Because $K-k<q=O(\sqrt{k})$ and $n\ge k$, we have $\npad=O(n)$.

The two replacements used below are worth keeping separate.  The inner hash family is obtained by deterministic search on the small domain $[q^3]$ and costs $2^{O(q\log q)}=2^{o(k)}$ time.  The outer hash family is assembled from three explicit splitter families and has $2^{o(K)}\log(2\npad)$ members.

\subsection{One deterministic hash-family lemma}

We first state the precise part of the $k$-restriction framework that we use.
In both applications below, $\Gamma$ is the set of assignments of hash values to one fixed $t$-element input set that make a product-automaton branch succeed.
The lemma turns the fraction of assignments that belong to $\Gamma$ into a deterministic family of hash functions such that every $t$-element input set receives an assignment in $\Gamma$.

\begin{lemma}[Deterministic realization of allowed hash patterns]
\label{lem:explicit-restriction}
Let $u,t,b$ be positive integers with $1\le t\le u$ and $b\le u$.  Let $\varnothing\ne\Gamma\subseteq[b]^t$ be closed under permutations of its coordinates: if $(x_1,\ldots,x_t)\in\Gamma$, then $(x_{\pi(1)},\ldots,x_{\pi(t)})\in\Gamma$ for every permutation $\pi$ of $[t]$. Put
\[
  \delta=\frac{|\Gamma|}{b^t}\ .
\]
Assume that membership in $\Gamma$ can be tested in time $\tau$.  Then one can deterministically construct a family $\mathcal R$ of functions $h:[u]\to[b]$ such that, for every $S=\{s_1<\cdots<s_t\}\subseteq[u]$, some $h\in\mathcal R$ satisfies
\[
  (h(s_1),\ldots,h(s_t))\in\Gamma\ .
\]
If $\delta=1$, one function suffices.  If $0<\delta<1$, the family has size at most
\begin{equation}\label{eq:restriction-size}
  \left\lceil
    \frac{t\ln u+1}{-\ln(1-\delta)}
  \right\rceil
  =O\!\left(\delta^{-1}t\log(2u)\right)\ .
\end{equation}
In either case it can be constructed in time
\begin{equation}\label{eq:restriction-time}
  O\!\left(
    \delta^{-1}\left(\frac{\e u^2}{t}\right)^t\tau
  \right)
  \le \delta^{-1}u^{2t+O(1)}\tau\ .
\end{equation}
\end{lemma}

\begin{proof}
Apply the one-constraint case of the $k$-restriction theorem of Naor, Schulman, and Srinivasan~\cite[Theorem~1, p.~185]{NaorSchulmanSrinivasan1995} with domain size $u$, restriction size $t$, alphabet size $b$, one constraint, and allowed set $\Gamma$.
Exactly $|\Gamma|$ of the $b^t$ possible assignments belong to $\Gamma$, so the fraction of allowed assignments is $\delta=|\Gamma|/b^t$. The theorem therefore gives a family of size at most the bound in \eqref{eq:restriction-size}.
Its construction time is
\[
  O\!\left(\delta^{-1}\binom ut\tau\,|\mathcal H_{u,t,b}|\right)\ ,
\]
where the explicit $t$-wise independent space used in that theorem satisfies $|\mathcal H_{u,t,b}|\le u^t$ because $b\le u$.
Using $\binom ut\le(\e u/t)^t$ gives \eqref{eq:restriction-time}.
\end{proof}

The four quantities in this lemma have direct meanings in our applications: $u$ is the domain of a hash, $t$ is the size of the input set that must be covered, $b$ is the number of hash values, and $\delta$ is exactly the success probability of one uniformly random hash on a fixed $t$-set.

\subsection{The hash family in Lemma~\ref{lem:cover}}

We first derandomize the family used in the fixed-histogram analysis.  Fix the gadget $G$, the number $N$ of copies, and the target length $k$ as in \Cref{sec:fixed-profile}.  Each hash maps an input symbol to a pair $(i,a)\in[N]\times C$, where $i$ selects one of the $N$ gadget copies and $a$ is a color of that copy.  Thus the hash has $cN$ possible values, so the parameter $b$ in \Cref{lem:explicit-restriction} is $b=cN$.  Let $\Gamma_{k,N}\subseteq([N]\times C)^k$ contain the ordered assignments with the following two properties:
\begin{enumerate}[label=(\alph*),leftmargin=3em]
  \item no copy-color cell occurs twice; and
  \item each copy receives at most $r$ cells, where $r$ is the capacity of $G$.
\end{enumerate}
The condition is unchanged when the $k$ coordinates are permuted.  Moreover,
\[
  |\Gamma_{k,N}|=k![z^k]B_G(z)^N\ ,
\]
so its density is the success probability in \eqref{eq:success}:
\[
  \frac{|\Gamma_{k,N}|}{(cN)^k}=p_{k,N}\ .
  \]
Membership is tested in $k^{O(1)}$ time by recording the used cells and the load of every copy.

Applying \Cref{lem:explicit-restriction} directly with domain $[n]$ would cost $n^{O(k)}$ time.  We therefore perform the standard size reduction first.
An $(u,t,\ell)$-splitter is a family of maps from $[u]$ to $[\ell]$ such that, for every $t$-element set $S\subseteq[u]$, some map assigns either $\lfloor t/\ell\rfloor$ or $\lceil t/\ell\rceil$ elements of $S$ to each value in $[\ell]$.  In particular, an $(u,t,t^2)$-splitter contains a map that is injective on every prescribed $t$-set.  Lemma~2 of Naor, Schulman, and Srinivasan~\cite[Lemma~2, p.~186]{NaorSchulmanSrinivasan1995} gives a globally explicit $(n,k,k^2)$-splitter $\mathcal A$ of size
\[
  |\mathcal A|=O(k^6\log k\log(2n))\ .
  \]

Now apply \Cref{lem:explicit-restriction} on the reduced domain $[k^2]$ with allowed patterns $\Gamma_{k,N}$.  In the fixed-histogram application $N=\Theta(k)$ and $c$ is fixed, so $cN\le k^2$ for all sufficiently large $k$; the finitely many remaining values can be handled directly.  This gives maps $r:[k^2]\to[N]\times C$ forming a family $\mathcal R$ of size
\[
  |\mathcal R|
  =O\!\left(p_{k,N}^{-1}k\log(2k)\right)\ .
\]
Define
\[
  \mathcal H_{\mathrm{fix}}
  \triangleq\{r\circ a:a\in\mathcal A,\ r\in\mathcal R\}\ .
  \]
For a $k$-set $S\subseteq[n]$, choose $a\in\mathcal A$ injective on $S$ and then choose $r\in\mathcal R$ that realizes an allowed pattern on the set $a(S)$.  Thus some member of $\mathcal H_{\mathrm{fix}}$ succeeds on $S$.
A smaller set is covered by extending it to a $k$-set, exactly as in \Cref{lem:cover}.  We have
\[
  |\mathcal H_{\mathrm{fix}}|
  \le p_{k,N}^{-1}k^{O(1)}\log(2n)\ ,
  \]
and the family can be listed in time complexity of
\[
  p_{k,N}^{-1}k^{O(k)}n^{O(1)}\ .
  \]
This already makes the fixed-histogram NFA a uniform fixed-parameter construction.  The factor $k^{O(k)}$ does not preserve the base displayed in \Cref{thm:profile}; preserving the final base requires the two-scale construction below.

\subsection{Explicit inner hashes}

We now derandomize the hashes in \Cref{prop:inner}.  Using the notation of that proposition: $G$ is the fixed partial gadget, $q$ is an allowed intermediate length, $N=\lambda q$ copies are composed, and the input alphabet for the resulting inner NFA is $[L]$, where later $L=q^3$.

Let $D=[N]\times C$ be the color set of the raw product.  Let $\Gamma_q\subseteq D^q$ contain the ordered $q$-tuples for which
\begin{enumerate}[label=(\alph*),leftmargin=3em]
  \item all $q$ copy-color cells are distinct; and
  \item in each copy, the set of local colors is certified by $G$.
\end{enumerate}
This is precisely the event counted in \eqref{eq:inner-success}.  Therefore
\[
  |\Gamma_q|=q![z^q]B_G(z)^N
  \qquad\text{and}\qquad
  \frac{|\Gamma_q|}{(cN)^q}=\rho\ .
  \]
The partial gadget is fixed, so its certified family is part of its constant-size description.  Membership in $\Gamma_q$ is consequently testable in $q^{O(1)}$ time.  For the Witt complete gadget the local test is simply $|S|\le6$.  For the eleven-Witt partial gadget $H_{11}$, a set is certified exactly when it uses at most six colors from each Witt palette and at most $61$ colors in total.

Apply \Cref{lem:explicit-restriction} with
\[
  u=L=q^3,
  \qquad t=q,
  \qquad b=cN,
  \qquad \delta=\rho\ .
\]
For large $q$, $b=O(q)\le L$.  We obtain a deterministic family of
\begin{equation}\label{eq:explicit-inner-family-size}
  O\!\left(\rho^{-1}q\log(2L)\right)
\end{equation}
inner hashes $h:[L]\to[N]\times C$.  Its construction time is
\[
  \rho^{-1}L^{2q+O(1)}=2^{O(q\log q)}\ .
  \]
To justify the last equality, property (P4) gives at least one certified product set of size $q$, hence $[z^q]B_G(z)^N\ge1$.  Since $q!\ge(q/\e)^q$ and $cN=c\lambda q$, equation \eqref{eq:inner-success} yields
\[
  \rho\ge(c\lambda\e)^{-q}\ .
\]
Thus $\rho^{-1}=2^{O(q)}$, while $L^{2q+O(1)}=2^{O(q\log q)}$.

We must also construct the compressed product, rather than merely count its states.  The raw product of $N=\Theta(q)$ copies of a fixed partial gadget has $2^{O(q)}$ states and ordinary transitions, so it can be generated and its unreachable states removed in $2^{O(q)}$ time.  Let the deleted band have width $s$, as in \Cref{lem:compress}.  For every reachable state $u$ at the lower boundary, compute all states $w$ reachable from $u$ after exactly $s$ raw transitions.  For every raw transition $w\xrightarrow{a}v$ into the retained upper boundary, add the compressed transition $u\xrightarrow{a}v$.

No search over hidden color sets is required.  Choose any path from the initial state to $u$.  If a continuation from $u$ repeated a color, either inside the continuation or from the chosen prefix, their concatenation would be a path of the raw product labeled by a repeated color.  This contradicts partial-gadget soundness, because every raw state is accepting.  Hence every path found by the layered reachability computation is a valid witness for the shortcut.
Even a cubic-time reachability computation in the raw product takes $2^{O(q)}$ time.

Finally, relabel the compressed product by every hash in \eqref{eq:explicit-inner-family-size} and take their nondeterministic union.
This deterministically lists the inner NFA $C_{q,L}$ from \Cref{prop:inner}.  Its preprocessing cost is $2^{O(q\log q)}$, in addition to time polynomial in the number of states and transitions that are output.

\subsection{Explicit outer hashes}\label{app:outer-hashes}

We first give the short probability calculation behind \Cref{lem:outer-hash} and then replace it by the standard explicit splitter construction.

\paragraph{Nonuniform existence proof}
Choose one function $g:[u]\to[T]\times[L]$ uniformly at random.  Its two coordinates play different roles.  The first coordinate maps the $K=qT$ labeled elements of a fixed $K$-set to $T$ blocks.  The fraction of assignments placing exactly $q$ elements in every block is
\[
  \rho_{\mathrm{bal}}
  \triangleq\frac{K!}{(q!)^T T^K}\ .
\]
Stirling's formula gives
\[
  -\log_2\rho_{\mathrm{bal}}=O(T\log q)=o(K)\ ,
\]
because $T=K/q$ and $q\to\infty$.

Condition on balanced first coordinates.  Inside one block, the probability that the $q$ second coordinates are distinct is
\[
  \frac{\fall{L}{q}}{L^q}
  =\prod_{i=0}^{q-1}\left(1-\frac{i}{L}\right)\ .
\]
The $T$ blocks are independent.  For $L=q^3$ and sufficiently large $q$, $i/L\le1/2$; hence $\log(1-z)\ge-2z$ gives
\[
  -\log\rho_{\mathrm{inj}}
  \le\frac{2T}{L}\sum_{i=0}^{q-1}i
  =O\!\left(\frac{K}{q^2}\right)
  =o(K)\ .
\]
Thus one random function succeeds on the fixed set with probability $\rho_{\mathrm{bal}}\rho_{\mathrm{inj}}=2^{-o(K)}$.  Taking $2^{o(K)}K\log(2u)$ independent functions and applying a union bound over the at most $u^K$ possible $K$-sets proves the existence claim.

\paragraph{Globally explicit construction}
We next replace the random family by published splitter and perfect-hash families.  Recall that
\[
  K=qT,
  \qquad L=q^3,
  \qquad q=\Theta(\sqrt K)\ ,
\]
and that the outer hashes are defined on the temporary alphabet $[\npad]$ from \eqref{eq:npad}.  We use three explicit splitter families from Naor, Schulman, and Srinivasan
\cite[Lemma~2 and Theorem~3(i),(iv), pp.~186--187]{NaorSchulmanSrinivasan1995}:
\begin{enumerate}[label=(\roman*),leftmargin=3em]
  \item an $(\npad,K,K^2)$-splitter $\mathcal A$, with
        \begin{equation}\label{eq:A-outer}
          |\mathcal A|=K^{O(1)}\log(2\npad)\ ,
        \end{equation}
        that contains a map injective on every prescribed $K$-set;
  \item a $(K^2,K,T)$-splitter $\mathcal B$, with
        \[
          |\mathcal B|
          =O\!\left(\frac{K^{2T+O(1)}\log K}{T!}\right)
          =2^{O(T\log K)}=2^{o(K)}\ ,
          \]
        that maps exactly $q$ elements of every prescribed $K$-set to each value in $[T]$; and
  \item a $(K^2,q,L)$-splitter $\mathcal P$, with
        \[
          |\mathcal P|=q^{O(1)}\log K\ ,
          \]
        that contains a map injective on every prescribed $q$-set.  Here $L=q^3\ge q^2$, so Theorem~3(iv) applies.
\end{enumerate}
These are globally explicit families: each family can be listed in time polynomial in its domain size and its own output size.

For $a\in\mathcal A$, $b\in\mathcal B$, and a tuple $\boldsymbol\phi=(\phi_1,\ldots,\phi_T)\in\mathcal P^T$, define
\begin{equation}\label{eq:explicit-outer-map}
  g_{a,b,\boldsymbol\phi}(x)
  \triangleq\bigl(b(a(x)),\ \phi_{b(a(x))}(a(x))\bigr)
  \in[T]\times[L]\ .
\end{equation}
Let $\mathcal G$ contain all maps in \eqref{eq:explicit-outer-map}.

\begin{lemma}[Explicit outer family]\label{lem:explicit-outer}
For every $K$-set $S\subseteq[\npad]$, some $g\in\mathcal G$ sends exactly $q$ elements of $S$ to each outer bucket and is injective on the second coordinates inside every bucket.  Moreover,
\begin{equation}\label{eq:explicit-outer-size}
  |\mathcal G|\le2^{o(K)}\log(2\npad)\ ,
\end{equation}
and the complete family can be listed in $2^{o(K)}\npad^{O(1)}$ time.
\end{lemma}

\begin{proof}
Fix $S\subseteq[\npad]$ with $|S|=K$. Choose $a\in\mathcal A$ injective on $S$. Choose $b\in\mathcal B$ such that
\[
  |a(S)\cap b^{-1}(i)|=q
  \qquad\text{for every }i\in[T].
\]
For each $i\in[T]$, choose $\phi_i\in\mathcal P$ injective on the $q$-element set $a(S)\cap b^{-1}(i)$. The resulting map $g_{a,b,\boldsymbol\phi}$ has the required two properties.

The number of maps is
\[
  |\mathcal G|=|\mathcal A|\,|\mathcal B|\,|\mathcal P|^T\ .
\]
Since $T=K/q=\Theta(\sqrt K)$,
\[
  \log_2|\mathcal B|=O(T\log K)=o(K)
\]
and
\[
  T\log_2|\mathcal P|
  =O\!\left(T(\log q+\log\log K)\right)=o(K)\ .
\]
Together with \eqref{eq:A-outer}, this proves \eqref{eq:explicit-outer-size}.  Enumerating the choices and evaluating each composed map on $[\npad]$ adds only a polynomial factor in $\npad$.
\end{proof}

\subsection{Uniform form of the main construction}

\begin{theorem}[Uniform amplification]\label{thm:uniform-amplification}
Under the hypotheses of \Cref{thm:amplification}, the NFA in that theorem can be constructed deterministically in
\begin{equation}\label{eq:uniform-amplification-time}
  2^{(E(x,y)+o(1))k}n^{O(1)}
\end{equation}
time.  The output consists of the complete state list and the complete list of ordinary one-symbol transitions.
\end{theorem}

\begin{proof}
For bounded $k$ the claim can be satisfied by a direct finite construction.
For larger $k$, choose the allowed value $q$, the padded value $K$, and $T=K/q$ exactly as in \Cref{sec:amplification}.  Thus $q=\Theta(\sqrt k)$ and $K=k+o(k)$.

First construct the inner NFA $C_{q,L}$, with $L=q^3$, by the deterministic procedure above.  Its extra preprocessing time is $2^{O(q\log q)}=2^{o(k)}$.  Next enlarge the alphabet temporarily to $[\npad]$, construct the family $\mathcal G$ from \Cref{lem:explicit-outer}, and for every $g\in\mathcal G$ list the raw product of $T$ copies of $C_{q,L}$.

Let $S_q=|Q(C_{q,L})|$ and $D_q=|\Delta(C_{q,L})|$. A branch for one outer hash has at most $S_q^T$ states. Store the outgoing local transitions of $C_{q,L}$ by local input symbol, and precompute $g(x)$ for every $x\in[\npad]$. Enumerating the product states takes $O(S_q^T)$ time. For each transition, choose the updated component, one of its local transitions, the states of the other $T-1$ components, and an original input symbol that induces the local transition. Hence all states and transitions of one branch can be listed in
\[
  O\!\left(S_q^T+\npad D_qS_q^{T-1}\right)
\]
time. The second term is also the transition bound used in the proof of \Cref{thm:amplification}. The nondeterministic union is formed without epsilon transitions by merging the initial states and taking the union of their outgoing transitions.

The state and transition estimates from \Cref{thm:amplification} therefore remain unchanged.  By \Cref{lem:explicit-outer}, the number of outer branches is only $2^{o(K)}\npad^{O(1)}$.  The total listing time is consequently
\[
  2^{(E(x,y)+o(1))K}\npad^{O(1)}\ .
\]
The construction of the inner family and all three outer splitter families is subexponential in $K$ and is absorbed by this bound.  Finally delete every transition labeled by one of the $K-k$ fresh symbols and retain only ranks $0,\ldots,k$.  Since $K=k+o(k)$ and $\npad=O(n)$, the result is \eqref{eq:uniform-amplification-time}.
\end{proof}

\begin{corollary}[Uniform form of the main theorem]\label{cor:uniform-main}
The NFA in \Cref{thm:main} can be constructed deterministically in
\[
  3.918^k n^{O(1)}
\]
time.
\end{corollary}

\begin{proof}
The eleven-Witt partial gadget $H_{11}$ is fixed.  Its finite compose-and-compress step is constructed once by the layered-reachability procedure used for the inner NFA.  Apply \Cref{thm:uniform-amplification} with the rational values $x$ and $y$ from \Cref{sec:eleven-witt}, and then use \eqref{eq:final-number}.
\end{proof}

If only rank $k$ is accepting, the same output graph recognizes the traditional exact-length language.  Ben-Basat, Gabizon, and Zehavi show that an explicit acyclic NFA of size $s$ for that language yields a deterministic minimum-weight simple $k$-path algorithm in time $O(sn^2\log W)$ for integer weights of magnitude at most $W$~\cite[Section~3]{BenBasatGabizonZehavi2016}.  Thus the uniform construction above also gives such an algorithm with running time $3.918^k n^{O(1)}\log W$.  This observation records the algorithmic consequence of uniformity; it is not used elsewhere in the paper.

\section{Coefficients of the eleven-Witt partial gadget}
\label{app:eleven-coefficients}

The next two tables expand \eqref{eq:Aeleven} and \eqref{eq:Beleven}.  They contain all exact integers used in the final substitution; no numerical search or external certificate is needed to verify the displayed bound.

\begin{table}[H]
\centering
\tiny
\setlength{\tabcolsep}{5pt}
\begin{tabular}{@{}r r r@{}}
\toprule
$j$ & $[z^j]A_{11}$ & $[z^j]B_{11}$\\
\midrule
0 & 1 & 1 \\
1 & 121 & 121 \\
2 & 7260 & 7260 \\
3 & 286891 & 287980 \\
4 & 8372595 & 8495410 \\
5 & 191926328 & 198792594 \\
6 & 3590148903 & 3843323484 \\
7 & 56226518975 & 63140310750 \\
8 & 750820249872 & 899749078800 \\
9 & 8664220979865 & 11296832619050 \\
10 & 87283865030813 & 126523976934740 \\
11 & 773660153828188 & 1276728613011390 \\
12 & 6070667004845190 & 11703086455223240 \\
13 & 42373390765351866 & 98121692243791370 \\
14 & 264118233584540500 & 756881012184576840 \\
15 & 1474717582817269098 & 5398390512819437562 \\
16 & 7395064892682784108 & 35756962984086163842 \\
17 & 33376621956399207090 & 220781732808203697900 \\
18 & 135843236767568489082 & 1275025954460637078700 \\
19 & 499444330218306720280 & 6907279173469073626150 \\
20 & 1661556681746456046794 & 35193441070815053735940 \\
21 & 5010046877807222954274 & 169041462406342601641390 \\
22 & 13714984061593772585100 & 767012962690500474869800 \\
23 & 34144119376323728018522 & 3293784641638203428118150 \\
24 & 77437025816050133825625 & 13408809122017810197265800 \\
25 & 160261587109715509988999 & 51823918771861501808788014 \\
26 & 303156976457717708620074 & 190409252871662685087548244 \\
27 & 524962205541521065357435 & 665843978673145925551573290 \\
28 & 833326882669836119261367 & 2218367455481676757416748920 \\
29 & 1214116241599850275349740 & 7048004806637136600063459390 \\
30 & 1625207726674894115232127 & 21370477843420614769784982696 \\
\bottomrule
\end{tabular}
\caption{Exact coefficients of $H_{11}$, degrees $0$ through $30$.}
\label{tab:aeleven-beleven-low}
\end{table}

\begin{table}[H]
\centering
\tiny
\setlength{\tabcolsep}{5pt}
\begin{tabular}{@{}r r r@{}}
\toprule
$j$ & $[z^j]A_{11}$ & $[z^j]B_{11}$\\
\midrule
31 & 1625207726674894115232127 & 61883160988698810044151119406 \\
32 & 1214116241599850275349740 & 171233699059034276254272577125 \\
33 & 833326882669836119261367 & 452968290930769116768808926075 \\
34 & 524962205541521065357435 & 1145955372380013863124868534200 \\
35 & 303156976457717708620074 & 2773354543347515465375527514310 \\
36 & 160261587109715509988999 & 6421766099824627001832732034410 \\
37 & 77437025816050133825625 & 14228056299440583813499827476760 \\
38 & 34144119376323728018522 & 30162378549580094655428298010380 \\
39 & 13714984061593772585100 & 61172550882732340044956092020660 \\
40 & 5010046877807222954274 & 118662668904783844745326962163008 \\
41 & 1661556681746456046794 & 220082143570904351703878914669128 \\
42 & 499444330218306720280 & 390088783091579848846523711227800 \\
43 & 135843236767568489082 & 660372807766656971235148843838400 \\
44 & 33376621956399207090 & 1066941562469213756447700967936800 \\
45 & 7395064892682784108 & 1643728762638690913078626400224480 \\
46 & 1474717582817269098 & 2412106515472011756720476125466880 \\
47 & 264118233584540500 & 3367408233588917777131894279473600 \\
48 & 42373390765351866 & 4465709629768930125562108491700800 \\
49 & 6070667004845190 & 5616109549756970662661370246585600 \\
50 & 773660153828188 & 6684380843158579956024573450993024 \\
51 & 87283865030813 & 7511917616038063290570634987749504 \\
52 & 8664220979865 & 7949016579455335405379951674583040 \\
53 & 750820249872 & 7894911698912352592545592448689920 \\
54 & 56226518975 & 7331500055535166639173837399632640 \\
55 & 3590148903 & 6336791071517629004643473671190016 \\
56 & 191926328 & 5069778473880165126986746880645376 \\
57 & 8372595 & 3729376854292201572403735168608000 \\
58 & 286891 & 2501473076022710659977161145907200 \\
59 & 7260 & 1513894347870970784258386016755200 \\
60 & 121 & 815471702059200619170645712949760 \\
61 & 1 & 383888224350711724813917287823360 \\
\bottomrule
\end{tabular}
\caption{Exact coefficients of $H_{11}$, degrees $31$ through $61$.}
\label{tab:aeleven-beleven-high}
\end{table}

\clearpage
\noindent\begin{minipage}{\textwidth}
\section{Interval verification of the final numerical bound}
\label{app:final-bound-verification}

This exact-arithmetic Python~3 script verifies the inequalities in
\Cref{sec:eleven-witt}. 
\begin{lstlisting}[style=pythoncopyable]
from fractions import Fraction as Q
from math import comb
from mpmath import iv

iv.dps = 60

def mul(p, q):
    r = [0] * (len(p) + len(q) - 1)
    for i, a in enumerate(p):
        for j, b in enumerate(q):
            r[i + j] += a * b
    return r

def power(p, n):
    r = [1]
    for _ in range(n):
        r = mul(r, p)
    return r

def evaluate(p, x):
    r = Q(0)
    for a in reversed(p):
        r = r * x + a
    return r

def I(x):
    return iv.mpf(x.numerator) / x.denominator

A_raw = power([1, 11, 55, 66, 55, 11, 1], 11)
B_raw = power([comb(11, j) for j in range(7)], 11)

# Raw ranks 31,...,35 are deleted; raw j >= 36 shifts to j-5.
A_11 = A_raw[:31] + A_raw[36:]
B_11 = B_raw[:62]

assert sum(A_raw[31:36]) == 10892645599457631372405834
assert sum(A_11) == 9587354400542368627594166

x, y = Q(153, 200), Q(81, 50)
Ax, By = evaluate(A_11, x), evaluate(B_11, y)
Bprime = [(j + 1) * B_11[j + 1] for j in range(61)]
mu = y * evaluate(Bprime, y) / By
lambda_ = 1 / mu

log2 = lambda z: iv.ln(z) / iv.ln(2)
H = log2(121 * I(lambda_) * iv.e) + log2(I(y)) - I(lambda_) * log2(I(By))
R = I(lambda_) * log2(I(Ax)) - iv.mpf("0.5") * log2(I(x))
E = H + R
base = iv.power(2, E)

assert H.b < iv.mpf("0.525678").a
assert R.b < iv.mpf("1.444240").a
assert E.b < iv.mpf("1.969918").a
assert base.b < iv.mpf("3.917459").a < iv.mpf("3.918").a

print("E =", E)
print("2^E =", base)
print("verified: 2^E < 3.917459 < 3.918")
\end{lstlisting}
\end{minipage}

\clearpage
\section{Proof of the compose-and-compress lemma}
\label{app:compress-proof}

\begin{proof}[Proof of \Cref{lem:compress}]
We verify the partial-gadget properties one at a time.
The raw product is a partial gadget by \Cref{lem:partial-product}; compose-and-compress changes only its state layers and the way paths cross the deleted band.

\paragraph{State polynomial}
The retained lower layers are the raw layers $0,\ldots,\ell$, so they contribute $\sum_{j=0}^{\ell}a_jz^j$.
A retained raw layer $j\ge\ell+s+1$ is shifted down by $s$, so it contributes $a_jz^{j-s}$.
Adding the two parts gives \eqref{eq:compressed-state}.

\paragraph{Soundness}
Take any path in the compressed automaton.
Every retained raw transition is also a transition of the raw product.
For each compressed jump, choose one raw path that witnesses its existence and replace the jump by that path.
The endpoints match by construction, so after all replacements we obtain one continuous path in the raw product.
If a visible color appeared twice on the compressed path, it would still appear twice on the expanded raw path.
This is impossible because the raw-product partial gadget accepts no word with a repeated color.
Hence every compressed path is repetition-free.

\paragraph{Completeness below the cut}
Let $S$ be a certified product set of size $j\le q$, and fix an arbitrary ordering $a_1a_2\cdots a_j$.
When $j\le\ell$, the raw accepting path has length at most $\ell$.
It stays entirely in the retained lower tail, so the same path exists in the compressed automaton.

\paragraph{Completeness across the cut}
Now suppose $j>\ell$.
The compressed path must cross from the lower tail to the upper tail after reading $a_1,\ldots,a_\ell$.
The deleted band has width $s$, so we need $s$ additional colors to witness that crossing.
By (P4), there is a certified set $T\supseteq S$ with $|T|=|S|+s$.
Write $T\setminus S=\{h_1,\ldots,h_s\}$; these colors are distinct and do not belong to $S$.
By (P3), the raw product accepts every ordering of $T$.
In particular, it accepts $a_1\cdots a_\ell h_1\cdots h_s a_{\ell+1}\cdots a_j$.
After the prefix $a_1\cdots a_\ell$, the next $s+1$ raw transitions read $h_1\cdots h_sa_{\ell+1}$.
The definition of compose-and-compress therefore adds a jump with the same endpoints and visible label $a_{\ell+1}$.
Replace those $s+1$ raw transitions by that jump.
The remaining suffix transitions are retained raw transitions.
The resulting compressed path reads exactly $a_1a_2\cdots a_j$, because the colors $h_1,\ldots,h_s$ are hidden rather than read.

\paragraph{Certified sets}
The raw product has $b_j$ certified sets of size $j$.
The argument above shows that every such set remains readable when $j\le q$.
We discard the larger certified sets because the new capacity is $q$.
This gives the truncation in \eqref{eq:compressed-good}.
Downward closure is preserved.
Property (P4) is also preserved: if a retained certified set has size $j$ and we request $h\le q-j$ additional colors, the product version of (P4) extends it to size $j+h\le q$, which is still retained.

\paragraph{Symmetry}
For $0\le j\le\ell$, the compressed layer of rank $j$ has size $a_j$.
Its mirror rank is $q-j$.
That rank comes from raw rank $q-j+s=R-j$ and therefore has size $a_{R-j}=a_j$, using the symmetry of the raw product.
Thus the compressed state profile is symmetric.
\end{proof}

\section*{Declaration of generative AI and AI-assisted technologies in the manuscript preparation process}
During preparation of this work, we used OpenAI's ChatGPT and Codex (models 5.6, 5.5, and 5.4 Pro) and Google's Gemini (model 3.1 Pro) to assist with proof verification, proof simplification, parameter optimization, literature survey, figure drawing, language editing, and organization. We independently reviewed and verified all mathematical arguments, citations, figures, code, and numerical calculations, edited the resulting material as needed, and takes full responsibility.

\AtBeginEnvironment{thebibliography}{%
  \renewcommand{\newline}{\unskip\space}%
  \def\urlprefix{}%
}
\bibliographystyle{elsarticle-num}
\bibliography{refs}

\end{document}